\documentclass[letterpaper, 11pt,nofootinbib,aps,pra,longbibliography]{revtex4-1}
\usepackage{microtype}

\usepackage{amsmath}
\usepackage{amssymb}
\usepackage{amsfonts}
\usepackage{amsthm}

\usepackage{color}
\definecolor{darkblue}{rgb}{0.,0.,0.4}
\definecolor{darkred}{rgb}{0.5,0.,0.}
\usepackage[pdftex,colorlinks=true,linkcolor=darkred,citecolor=darkblue,urlcolor=blue]{hyperref}

\DeclareMathOperator{\coker}{\mathrm{coker}}
\DeclareMathOperator{\im}{\mathrm{im}}
\DeclareMathOperator{\ann}{\mathrm{ann}}

\newcommand{\CC}{\mathbb{C}}
\newcommand{\ZZ}{\mathbb{Z}}
\newcommand{\id}{\mathrm{id}}
\newcommand{\FF}{\mathbb{F}} 
\newcommand{\mm}{\mathfrak{m}} 

\newcommand{\ket}[1]{\left| {#1} \right\rangle}

\newcommand{\braket}[2]{\langle {#1} | {#2} \rangle}

\newtheorem{lem}{Lemma}[section]
\newtheorem{thm}[lem]{Theorem}
\newtheorem{cor}[lem]{Corollary}
\newtheorem{prop}[lem]{Proposition}

\theoremstyle{definition}
\newtheorem{rem}[lem]{Remark}

\newtheorem{exc}[lem]{Excercise}
\theoremstyle{plain}

\begin{document}

\title{Algebraic Methods for Quantum Codes on Lattices}
\author{Jeongwan Haah}
\date{17 October 2016}
\affiliation{
Department of Physics, Massachusetts Institute of Technology, Cambridge, Massachusetts, USA}
\affiliation{
Station Q Quantum Architectures and Computation Group, Microsoft Research, Redmond, Washington, USA}
\email{jwhaah@microsoft.com}

\begin{abstract}
This is a note from a series of lectures
at Encuentro Colombiano de Computaci\'on Cu\'antica, 
Universidad de los Andes, Bogot\'a, Colombia, 2015.
The purpose is to introduce additive quantum error correcting codes,
with emphasis on the use of binary representation of Pauli matrices
and modules over a translation group algebra.
The topics include symplectic vector spaces,
Clifford group, cleaning lemma,
an error correcting criterion,
entanglement spectrum,
implications of the locality of stabilizer group generators,
and 
the classification of translation-invariant one-dimensional additive codes
and two-dimensional CSS codes with large code distances.
In particular,
we describe an algorithm to find a Clifford quantum circuit (CNOTs)
to transform any two-dimensional translation-invariant
CSS code on qudits of a prime dimension with code distance being the linear system size,
into a tensor product of finitely many copies of the qudit toric code and a product state.
Thus, the number of embedded toric codes is the complete invariant of these CSS codes
under local Clifford circuits.
\end{abstract}

\maketitle
\tableofcontents

\section{Introduction}

Quantum error correcting codes were invented by Shor who showed
that a quantum computer can in principle be built
out of faulty components~\cite{Shor1995nine, Shor1996Fault-tolerant}.
The basic idea is that despite 
quantum mechanical state vectors and operators form
continuous spaces,
errors can be effectively treated as if they were discrete.
Before long,
a simple and systematic method of slicing the state vector space
into discretely labeled subspaces with help of classical error correcting codes
was discovered by Calderbank and Shor~\cite{CalderbankShor1996Good} 
and Steane~\cite{Steane1996Multiple},
and generalized to what is now known as 
stabilizer codes by Gottesman~\cite{Gottesman1996Saturating}
or symplectic/additive codes by Calderbank, Rains, Shor and 
Sloane~\cite{CalderbankRainsShorEtAl1997Quantum,CalderbankRainsShorSloane1998GF4}.

Another conceptual and practical use of quantum error correcting codes
is provided by Kitaev~\cite{Kitaev2003Fault-tolerant},
who presented a class of exactly solvable local Hamiltonians
associated with quantum error correcting codes,
exhibiting so-called topological order.
Although the phenomenology of a certain topological order
was known by Sachdev and Read~\cite{ReadSachdev1991LargeN},
Kitaev's models facilitated understanding significantly,
and demonstrated how a topologically ordered medium
can be used as a naturally fault-tolerant quantum information processing platform.

In this note,
we present a mathematically coherent and mostly self-contained treatment
of stabilizer/additive/symplectic codes
with applications to quantum spin systems
governed by a translation-invariant local Hamiltonian associated with codes.
We emphasize binary symplectic vector spaces over groups of Pauli matrices.
This makes the appearance of translation-group algebra with coefficients
in the binary field very natural,
which we will mainly study in the later half of the present note.

We discuss neither 
a particular way of designing a quantum error correcting code
nor decoding algorithms thereof.
Also, we do not attempt to develop insight about topological order in general.
Rather, this note is to introduce and review notions from commutative algebra
that the author has found useful and interesting
in the understanding of the cubic code model~\cite{Haah2011Local} and its 
cousins~\cite{
Chamon2005Quantum,
BravyiLeemhuisTerhal2011Topological,
Kim2012qupit,
Yoshida2013Fractal,
VijayHaahFu2015New,
VijayHaahFu2016Fracton}.
Most results in this note is hardly new,
but the derivation of the results
will be sometimes different from existing literature.
The classification theorem \ref{CSS2dClassification}
for two-dimensional translation-invariant CSS
codes has not previously appeared,
and is intimately related to a result of Bomb\'in~\cite{Bombin2011Structure}.
Detailed comparison is given in \ref{BombinResult}.

The author thanks
Cesar Galindo-Martinez and Julia Plavnik
for their hospitality during the workshop in Bogot\'a, Colombia,
and H\'ector Bomb\'in for guiding along his paper~\cite{Bombin2011Structure}.
The author was supported by Pappalardo Fellowship in Physics while at MIT.

\section{Additive/Stabilizer/Symplectic codes}

The set of all $2 \times 2$ matrices acting on $\CC^2 = \mathrm{span}~\{ \ket 0, \ket 1 \}$
has a linear basis consisting of
\begin{align}
 I = \begin{pmatrix} 1 & 0 \\ 0 & 1 \end{pmatrix}, \quad
 X = \begin{pmatrix} 0 & 1 \\ 1 & 0 \end{pmatrix}, \quad 
 Y = \begin{pmatrix} 0 & -i \\ i & 0 \end{pmatrix}, \quad
 Z = \begin{pmatrix} 1 & 0 \\ 0 & -1 \end{pmatrix},
\end{align}
called {\bf Pauli matrices}. They square to be the identity, are hermitian, and satisfy
\begin{align}
 XY = iZ,\quad YZ = iX,\quad ZX = iY.
\end{align}
Thus, any pair of non-identity Pauli matrices anti-commute.
As a complex algebra, we can further reduce the generating set to $\{X,Z\}$ since $Y = -iZX$.
The matrix algebra on $(\CC^2)^{\otimes n}$ is the $n$-fold tensor product of the
$2 \times 2$ matrix algebra, of which the generating set can be chosen as
the set of all $n$-fold tensor products of the Pauli matrices.
To avoid lengthy phrasing, let us say just {\bf Pauli operators} 
to mean $n$-fold tensor products of Pauli matrices.

An {\bf additive code}~\cite{CalderbankRainsShorSloane1998GF4,CalderbankRainsShorEtAl1997Quantum}
or {\bf stabilizer code}~\cite{Gottesman1996Saturating}
is a subspace of $n$-qubit Hilbert space%
\footnote{In functional analysis, 
a Hilbert space refers to an inner product space
that is complete in the induced metric topology.
For finite dimensional vector spaces,
the completeness, meaning that every Cauchy sequence converges,
follows trivially from the completeness of the real numbers.
So, the wording ``Hilbert space'' is superfluous in our setting
where we only consider finite dimensional spaces.
Nonetheless, we will keep using this terminology
whenever we are referring to the complex vector space of all state vectors.
}
that is defined as the common eigenspace of a set of commuting Pauli operators of eigenvalue $+1$.
The common eigenspace is referred to as {\bf code space} to distinguish it from other related entities.
Any vector in the code space is a {\bf code vector}.
The defining Pauli operators or any product of them are called {\bf stabilizers}.
The stabilizers form the multiplicative {\bf stabilizer group}.

\begin{exc}
Show that the stabilizer group of nonzero additive code space does not contain $-1$.
\label{NoMinusOne}
\hfill $\diamond$ \end{exc}

How do we test whether a pair of Pauli operators commute?
Since $X$ and $Z$ generate the matrix algebra,
we see that any Pauli operator $P$ can be written, for example, as
\begin{align}
 P = \eta 
 \underbrace{(X \otimes I \otimes X \otimes \cdots \otimes X )}_{x = 101\cdots 1}
 \underbrace{(Z \otimes Z \otimes I \otimes \cdots \otimes I )}_{z = 110\cdots 0}
 = \eta X(x) Z(z)
\end{align}
where $\eta = \pm 1 , \pm i$, a fourth root of unity.
We have associated bit strings $x$ and $z$ to keep record of the positions of nontrivial tensor factors.
It is now clear that any Pauli operator is uniquely 
specified by the overall phase factor $\eta$ and two bit strings of length $n$.
The commutation relation between a pair of Pauli matrices is then
calculated as
\begin{align}
&\eta_1 X(x_1) Z(z_1) \eta_2 X(x_2) Z(z_2) \nonumber \\ 
 &= \eta_1 \eta_2 X(x_1)\left[ (-1)^{ x_2 \cdot z_1 } X(x_2) Z(z_1) \right] Z(z_2) \nonumber \\
 &= \eta_1 \eta_2 (-1)^{ x_2 \cdot z_1 } X(x_1 + x_2 \mod 2) Z(z_1+ z_2 \mod 2) \label{eq:bitaddition}\\
 &= \eta_1 \eta_2 (-1)^{ x_2 \cdot z_1 }X(x_2)X(x_1) Z(z_2) Z(z_1)\nonumber \\
 &= \eta_1 \eta_2 (-1)^{ x_2 \cdot z_1 }X(x_2)\left[ (-1)^{x_1 \cdot z_2} Z(z_2) X(x_1) \right] Z(z_1) \nonumber \\
 &= (-1)^{ x_2 \cdot z_1 - x_1 \cdot z_2} \eta_2 X(x_2)Z(z_2) \eta_1 X(x_1)Z(z_1).\label{eq:commuted}
\end{align}
Hence, any pair of Pauli operators commute or anti-commute,
and the two cases are distinguished by 
\begin{align}
 -x_2 \cdot z_1 + x_1 \cdot z_2 \mod 2.
\end{align}
More generally,
we can consider $d \times d$ matrix algebra acting on $\CC^d$ with generators
\begin{align}
 X_d =
 \begin{pmatrix}
  0 & 1 &  &  & \\
    & 0 & 1&  & \\
  \vdots &   & 0&1 & \\
  0  &   &  &  & \ddots \\
  1 & 0 & \cdots & &
 \end{pmatrix}, \quad
Z_d =
\begin{pmatrix}
 1 & & & & \\
 & \omega & & & \\
 & & \omega^2 & &\\
 & & & \ddots & \\
 & & & & \omega^{d-1}
\end{pmatrix}
\label{eq:XZdMatrices}
\end{align}
where $\omega = \exp( 2\pi i / d)$.
Due to the commutation relation
\begin{align}
 X_d Z_d = \omega Z_d X_d ,
 \label{eq:XZdCommutation}
\end{align}
any pair of these generalized Pauli operators $X_d(x)Z_d(z)$ and $X_d(x') Z_d(z')$
commutes up to $\omega^m$ where
\begin{align}
 m = -x' \cdot z + x \cdot z'  \pmod d 
 = \begin{pmatrix} x & z \end{pmatrix}
 \begin{pmatrix} 0 & \id \\ -\id & 0 \end{pmatrix}
 \begin{pmatrix} x' \\ z' \end{pmatrix}
 \pmod d .
\label{eq:symp}
\end{align}
When $x=x'$ and $z=z'$, we trivially have $m = 0$.

\begin{exc}
Verify \eqref{eq:symp} by deriving the analogue of \eqref{eq:commuted}.
\hfill $\diamond$ \end{exc}

\begin{exc}
Let $G$ be the group of $3 \times 3$ matrices of form
\begin{align}
h_{a,b,c} = \begin{pmatrix}
1 & a & c \\ 0 & 1 & b \\ 0 & 0 & 1
\end{pmatrix}
\end{align}
where $a,b,c \in \ZZ/d\ZZ$ and the group operation is the matrix multiplication.
Show that the matrices in \eqref{eq:XZdMatrices} together with $\omega = \exp(2\pi i / d)$
form a representation of the group $G$.
(The group $G$ has a name, {\bf Heisenberg group} over $\ZZ/d\ZZ$.)
{\it Hint}: What is the commutator of $h_{1,0,0}$ and $h_{0,1,0}$?
What does $h_{0,0,1}$ correspond to?
\hfill $\diamond$ \end{exc}

We may regard the {\em dit} string $x,z$ 
as one dit string of length $2n$~\cite{Gottesman1999qudit,Rains1999Nonbinary}.
If $d$ is a prime number, then we may further regard the dit string of length $2n$ 
as a $2n$-dimensional vector over the finite field $\FF_d$.%
\footnote{At this stage, it is less clear why we need the field rather than just additive group $\ZZ / d\ZZ$.
See \eqref{eq:bitaddition}.
It is actually more of a technical convenience rather than an essential ingredient.
However, some of our claims we will make below depend on the fact that $\ZZ/d\ZZ = \FF_d$ is a field.
See Ref.~\cite{Farinholt2013} for discussions regarding composite numbers $d$.
}
Upon multiplication of two Pauli operators, the corresponding dit string is added modulo $d$,
which can be interpreted as the vector addition.
In addition, we see that Eq.~\eqref{eq:symp} introduces a symplectic form on this vector space.
This symplectic form will be central to further development,
and it is thus necessary to understand the symplectic structure thoroughly.

\subsection{Symplectic vector spaces}

Let $\FF$ be any field and $V$ denote a vector space over $\FF$.
A bilinear form $\lambda : V \times V \to \FF$ is called {\bf symplectic}
or {\bf alternating} if 
\begin{align}
\lambda( v, v ) = 0 \quad \forall v \in V .
\end{align}
If the association
\begin{align}
V \ni v \mapsto \lambda(v, \cdot ) \in V^*
\end{align}
from $V$ to its dual vector space $V^*$
is bijective,
then we say $\lambda$ is {\bf non-degenerate}.
When $V$ is finite dimensional, we can express the $\lambda$ as a matrix $\Lambda$ given a basis of $V$,
and the bilinear form $\lambda$ is non-degenerate if and only if the matrix $\Lambda$ has a nonzero
determinant.
Below we will not distinguish the form $\lambda$ from its matrix representation $\Lambda$
whenever the basis choice is clear.

\begin{exc}
Show that the matrix representation of the symplectic form is skew-symmetric ($\Lambda^T = - \Lambda$),
but a skew-symmetric matrix may not yield a symplectic form if the field is of characteristic 2.
\hfill $\diamond$ \end{exc}

In order to understand the symplectic space better,
we will find a canonical basis.
To this end, we consider a variant of Gram-Schmidt orthogonalization
for inner product spaces.
Let $V$ be a $n$-dimensional symplectic space, not necessarily non-degenerate.
With respect to $\lambda$, we can consider the symplectic complement 
\begin{align}
W^\perp = \{ v \in V ~:~ \lambda(v,w) = 0 \quad \forall w \in W \}
\end{align}
of any subspace $W$.
(Some authors call it as the ``orthogonal'' complement.)
\begin{lem}
If $v,w \in V$ satisfy $\lambda(v,w) \neq 0$,
then 
\begin{align}
V = \mathrm{span}\{ v, w\} \oplus (\mathrm{span}\{v,w\})^\perp .
\end{align}
\end{lem}
\begin{proof}
``$\supseteq$'' is trivial by definition.
Let $v' = v$ and $w' = w/ \lambda(v,w)$ form a basis for $W =\mathrm{span} \{v,w\}$
so that $\lambda(v',w') = 1$.
The decomposition
\begin{align}
x = -\lambda(x,v')w' + \lambda(x,w')v' + \underbrace{x + \lambda(x,v')w' - \lambda(x,w')v'}_{y}
\end{align}
for an arbitrary $x \in V$ proves ``$\subseteq$'' since
\begin{align}
\lambda(y,v') &= \lambda(x,v') + \lambda(x,v')\lambda(w',v') - \lambda(x,w')\lambda(v',v') = 0,\\
\lambda(y,w') &= \lambda(x,w') - \lambda(x,v')\lambda(w',w') - \lambda(x,w')\lambda(v',w') = 0.
\end{align}
If $z \in W \cap W^\perp$, then $z = a v' + b w'$ for some $a,b \in \FF$ 
and $\lambda(z,v') = \lambda(z,w') = 0$,
which imply $a = 0$ and $b = 0$.
Therefore, the sum $W + W^\perp$ is direct.
\end{proof}
A two-dimensional subspace on which the symplectic form is non-degenerate,
as in the lemma,
is called a {\bf hyperbolic plane}.
Note that the matrix representation of the symplectic form for the hyperbolic plane is
\begin{align}
\begin{pmatrix}
0 & 1 \\ -1 & 0 
\end{pmatrix} .
\end{align}

Suppose we have an unstructured basis $\{v_1, \ldots, v_n \}$ where $\lambda$ is not always zero.
By examining all values $\lambda(v_i,v_j)$,
we can find a hyperbolic plane.
By the lemma, the symplectic complement of the span of $\{ v_i , v_j \}$ has smaller dimension,
and we can inductively proceed to decompose the space.
At some point the decomposition may encounter a
subspace on which the symplectic form vanishes,
which is called an {\bf isotropic} subspace.
The proof of the lemma gives an algorithm to find a canonical basis,
which is essentially the same as the Gram-Schmidt
orthogonalization for inner product spaces.
We arrive at a structure theorem of finite dimensional symplectic spaces.
\begin{prop}
Any finite dimensional symplectic vector space is a direct sum of hyperbolic planes
and an isotropic subspace.
In particular, a non-degenerate symplectic vector space is even dimensional.
\label{SymplecticSpaceStructure}
\end{prop}
\begin{exc}
Show that the dimension of an isotropic subspace of a non-degenerate symplectic space of dimension $2n$
is at most $n$.
\hfill $\diamond$ \end{exc}

\subsection{Automorphisms of symplectic spaces}

By definition, an automorphism $A$ of a symplectic vector space
$V$ is an invertible linear map from $V$ to itself such that
\begin{align}
\lambda( Av, Aw ) = \lambda( v,w) \quad \forall v,w \in V .
\label{PreservingSymplecticForm}
\end{align}
With respect to a basis of $V$
 the condition for the matrix $A$ to be an automorphism is 
\begin{align}
A^T \lambda A = \lambda, \quad \det A \neq 0 
\end{align}
where the superscript $T$ denotes the transpose.
If $\lambda$ is non-degenerate, the second condition $\det A \neq 0$ is redundant.

Let us choose a canonical basis of $2n$-dimensional non-degenerate symplectic vectors space where
the matrix representation of the symplectic form is
\begin{align}
\lambda_n =
\begin{pmatrix}
0 & \id \\ -\id & 0
\end{pmatrix}
\label{eq:symplecticMatrix}
\end{align}
where $\id$ stands for the $n \times n$-identity matrix.
Thus, for the unit column vectors $e_i$, where the sole nonzero 1 appears at $i$-th component $(i = 1,\ldots, n)$,
we have $\lambda(e_i, e_{i+n}) = 1 = - \lambda(e_{i+n},e_i)$,
and all the other symplectic pairings vanish.

A few elementary automorphisms can be found by solving an equation,
\begin{align}
&\begin{pmatrix}
a & b \\ c & d
\end{pmatrix}^T
\begin{pmatrix}
0 & 1 \\ -1 & 0
\end{pmatrix}
\begin{pmatrix}
a & b \\ c & d
\end{pmatrix}
=
\begin{pmatrix}
0 & 1 \\ -1 & 0
\end{pmatrix} \Leftrightarrow 
ad-bc = 1.
\end{align}
Thus, the automorphism group for a hyperbolic plane is $SL(2,\FF)$,
which is generated by three types of elements
\begin{align}
S = \begin{pmatrix}
 1 & 0 \\ a & 1
\end{pmatrix}
\text{ where } a \in \FF, \quad
H = \begin{pmatrix}
 0 & 1 \\ -1 & 0
\end{pmatrix},\quad \text{ and }
R = \begin{pmatrix}
 a & 0 \\ 0 & a^{-1}
\end{pmatrix}
\text{ where } a \in \FF^\times .
\label{ElementaryZeroDSymplecticTransformation}
\end{align}

If $n \ge 2$, we see that there is another automorphism $C$
\begin{align}
 C|_{\langle e_i , e_j , e_{i+n} , e_{j+n} \rangle } = 
\begin{pmatrix}
 1 & 0 & & \\
 a & 1 & & \\
  &  & 1 & -a \\
  &  & 0 & 1
\end{pmatrix} \text{ where } a \in \FF, ~~1 \le i \neq j \le n
\label{CNOT}
\end{align}
where we only displayed the action of $C$ on the four dimensional subspace,
on the complement of which $C$ acts by the identity.
It is often useful to think of the four elementary {\bf symplectic transformations}
(automorphisms)
as row operations on a column vector.
$S$ adds $a$-times $i$-th component to $(i+n)$-th component.
$H$ interchanges $i$-th and $(i+n)$-th with an extra sign.
$C$ adds $a$-times $i$-th component to $j$-th component ($j \neq i \le n$)
while $(-a)$-times $(n+j)$-th to $(n+i)$-th.

\begin{prop}
Let $\Sigma$ be an $s$-dimensional isotropic subspace 
of a $2n$-dimensional non-degenerate symplectic space $\FF^{2n}$.
There exists a symplectic transformation that maps $\Sigma$ 
onto the span of $\{e_1,\ldots, e_s\}$, a subset of the canonical basis vectors of $\FF^{2n}$.
Moreover, the symplectic complement $\Sigma^\perp$ is isomorphic to
$\Sigma \oplus W$ for some non-degenerate symplectic subspace $W$.
\label{IsotropicSubspaceComplement}
\end{prop}
\begin{proof}
We provide an algorithm to find the symplectic transformation (automorphism)
using $S$, $H$, $R$, and $C$.
Let us identify the subspace $\Sigma$ with a matrix of basis vectors written in the columns.
This matrix is $2n \times s$.
Any column operation on the matrix $\Sigma$ is just a different choice of basis vectors
for the space $\Sigma$.
If we could find a combination of the four elementary symplectic transformations
and column operations
such that the matrix $\Sigma$ is transformed to 
\begin{align*}
\begin{pmatrix}
\id & 0 \\
0 & 0 \\
\hline
0 & 0 \\
0 & 0
\end{pmatrix},
\end{align*}
then the composition of the elementary symplectic transformations
is the desired transformation.
We inserted a horizontal line to distinguish upper and lower half blocks.

Using $C$, we see that any row operation on the upper-half block can be made
symplectic by a suitable row operation on the bottom-half block.
Hence, for the upper half-block, we can freely employ row and column operations
to the matrix $\Sigma$ to obtain
\begin{align*}
\Sigma' = \begin{pmatrix}
1 & 0_{1 \times (s-1)} \\
0_{(n-1) \times 1} & \star_{(n-1) \times (s-1)} \\
\hline
a & \star \\
b_{(n-1)\times 1} & \star
\end{pmatrix}.
\end{align*}
Then, $S$ on the 1st and $(n+1)$st row can make the entry $a$ to zero.
$H$'s on the all but 1st and $(n+1)$st row bring the submatrix $b$
to the upper half block, and a subsequent row operation can annihilate it.
As a result, we obtain
\begin{align*}
\Sigma'' = \begin{pmatrix}
1 & 0_{1 \times (s-1)} \\
0_{(n-1) \times 1} & \star_{(n-1) \times (s-1)} \\
\hline
0 & c_{1 \times (s-1)} \\
0_{(n-1)\times 1} & \star
\end{pmatrix}.
\end{align*}
Being isotropic, it satisfies $(\Sigma'')^T \lambda_n \Sigma'' = 0$
where $\lambda_n$ is in \eqref{eq:symplecticMatrix}.
This equation tells us that the submatrix $c$ has to be zero.
We have reduced the dimension in the problem: $s \to s-1, n \to n-1$.
The desired transformation is found by recursion.

The second statement is a corollary of the first
since the symplectic complement of $\{e_1,\ldots,e_s\}$
is the span of $\{e_1, \ldots, e_s\}$ plus the span of $\{e_{s+1},\ldots, e_{n}, e_{n+s+1},\ldots,e_{2n}\}$,
where the latter is non-degenerate.
\end{proof}

\begin{prop}
$S$, $H$, $R$, and $C$ generate the full automorphism group of a finite dimensional 
non-degenerate symplectic space.
\label{SymplecticTransformationGenerators}
\end{prop}
\begin{exc}
Prove \ref{SymplecticTransformationGenerators}. 
{\em Hint}:
Transform the matrix of an automorphism
into the identity matrix by $S$, $H$, $R$, and $C$,
using the similar strategy as in the proof of \ref{IsotropicSubspaceComplement}.
\hfill $\diamond$ \end{exc}

\subsection{Logical operators}

We return to the discussion of Pauli operators.
We have seen that the multiplicative group of Pauli operators
is rather simple, since, if we ignore the overall phase factors $\pm 1, \pm i$,
the group is actually abelian. More precisely,
\begin{prop}
If $\mathcal P$ denotes the multiplicative group of all Pauli operators including $\pm 1, \pm i$
acting on $n$-qubit Hilbert space $(\mathbb C^2)^{\otimes n}$,
then $\mathcal P / \langle i \rangle$ is an abelian group, which is isomorphic to the additive group 
$(\ZZ/2\ZZ)^{2n}$.
\end{prop}
\noindent
It is a useful coincidence that this additive group can be regarded as a vector space $\FF_2^{2n}$ over $\FF_2$.
Moreover, the commutation relation naturally endows this vector space
with a symplectic form.
The form is non-degenerate because $X$ and $Z$ on $i$-th qubit defines a hyperbolic plane,
the direct sum of which is the whole space.
We use by convention an ordered canonical basis $e_1,\ldots, e_{2n}$ on the symplectic space $\FF_2^{2n}$
in which the symplectic matrix is as in \eqref{eq:symplecticMatrix}.
So, the Pauli operator $P(e_i)$ corresponding to $e_i$ is equal up to a phase factor to
\begin{align}
 P(e_i) = \begin{cases} X_i & \text{ if } 1 \le i \le n, \\ Z_i &\text{ if } n+1 \le i \le 2n \end{cases}
\end{align}
where the subscript $i$ denote the sole qubit that is acted on nontrivially by the designated Pauli matrix.

We now introduce the simplest additive code on $n$ qubits.
Let $\{ Z_i : i = 1,\ldots, n-k \}$ be a commuting set of Pauli operators.
What is the code space, the common ($+1$)-eigenspace of the stabilizers?
It is obvious that any such eigenvector (code vector) must be of form
\begin{align}
\ket 0 \otimes  \cdots \otimes \ket 0 \otimes \ket \psi
\label{eq:trivialCodeVectors}
\end{align}
for some vector $\ket \psi \in (\CC^2)^{\otimes k}$,
and we can identify the code space with this $(\CC^2)^{\otimes k}$.
There exist Pauli operators $X_{n-k+1}, Z_{n-k+1},\ldots, X_n, Z_n$
that generate the operator algebra on the code space.
These Pauli operators are called logical operators.%
\footnote{
The term ``logical'' comes from the intended use of the code in an error correction scheme,
where information is redundantly encoded into a physical system
shielding the ``logical level'' from errors, and the logical operators are those
that transform the encoded information.
Of course, in this oversimplified example, there is no protection.
}
The logical operators may appear in various guises.
For example, $Z_1 Z_n$ has the same action on the code space as $Z_n$ since
$(Z_1 Z_n) (Z_n^{-1}) = Z_1$ acts by the scalar $+1$ on the code space
by construction.
Formally, any Pauli operator that maps the code vector into the code space 
is called a {\bf logical operator}.
The logical operator $P$ is said to be {\bf equivalent} to another logical operator $P'$
if $P'P^{-1}$ is a stabilizer up to a phase factor.
Sometimes, the logical operator is termed as a {\em nontrivial logical operator}
in order to distinguish it from a {\em trivial logical operator},
where the latter is nothing but a stabilizer.
\begin{exc}
Show that a Pauli operator is a (trivial or nontrivial) logical operator
if and only if it commutes with every stabilizer.
\label{LogicalOperatorsCommutant}
\hfill $\diamond$ \end{exc}

Let us translate the discussion around the simplest example into the language of symplectic spaces.
The stabilizers generate an abelian multiplicative group, called {\bf stabilizer group} $\mathcal S$.
Since it is abelian, the corresponding symplectic space is isotropic.
By \ref{LogicalOperatorsCommutant}, we see that the set $\mathcal L$
of all trivial and nontrivial logical operators
corresponds to the symplectic complement of this isotropic subspace: $\mathcal L = \mathcal S^\perp$.
By \ref{IsotropicSubspaceComplement}, $\mathcal S^\perp$
contains $\mathcal S$ and a non-degenerate subspace which is isomorphic to $\mathcal S^\perp / \mathcal S$.
Any nonzero element in this quotient space corresponds to a nontrivial logical operator.
Two different logical operators may have the same action on the code space,
which is precisely captured by the quotient space $\mathcal S^\perp / \mathcal S$.
All these statements on the symplectic space does not rely on a specific basis.
What does the basis change (automorphism) in the symplectic space over $\FF_2$ corresponds to
in the space of operators on the Hilbert space over $\CC$?

\subsection{Clifford group}

A unitary operator on $n$-qubit Hilbert space (over $\CC$) is called {\bf Clifford} 
if it maps any Pauli operator to a Pauli operator.
Any Clifford operator $U$ induces a linear map $A$ in the corresponding symplectic space.
This is easily proved as
\begin{align*}
U P( v + w) U^\dagger &\propto UP(v)U^\dagger UP(w)U^\dagger,\\
\pm P( A(v+w) ) &\propto P(A(v)) P(A(w)) \propto P(A(v) + A(w))
\end{align*}
where $P(v)$ denotes a Pauli operator on $n$ qubits
specified by the bit string of length $2n$.
Since $U$ is invertible, the induces linear map $A$ is also invertible.
Recall that the symplectic form $\lambda$ is defined by
\begin{align}
 P(v)P(w) = (-1)^{\lambda(v,w)} P(w) P(v).
\end{align}
Conjugating this equation by $U$,
we see that $A$ preserves the symplectic form in the sense of \eqref{PreservingSymplecticForm}.
Therefore, a Clifford unitary induces an automorphism of the symplectic space.

We have found a generating set of the symplectic group on $\FF_2^{2n}$.
Are the generators induced from Clifford unitary operators?
Consider the following unitaries.
\begin{align}
 U_H = \frac{1}{\sqrt{2}} \begin{pmatrix} 1 & 1 \\ 1 & -1 \end{pmatrix}, \quad
 U_S = \begin{pmatrix} 1 & 0 \\ 0 & i \end{pmatrix}, \quad
 U_{CNOT} = \begin{pmatrix} 1 & 0 & & \\ 0 & 1 & & \\ & & 0 & 1 \\ & & 1 & 0 \end{pmatrix}.
\label{CliffordGenerators}
\end{align}
They are called Hadamard gate, phase gate, and controlled-not gate, respectively.
It is straightforward to verify that $U_H$ interchanges $X$ and $Z$ by conjugation.
Therefore, $U_H$ induces the elementary symplectic transformation $H$ of \ref{SymplecticTransformationGenerators}.
\begin{exc}
Show that $U_S$ and $U_{CNOT}$ are Clifford, 
and induce $S$ and $C$ of \ref{SymplecticTransformationGenerators}, respectively.
{\em Hint}: A linear operator is determined by the image of basis vectors.
\hfill $\diamond$ \end{exc}
By \ref{SymplecticTransformationGenerators}, it follows that any symplectic automorphism
is induced by some Clifford unitary operator. This is summarized by stating
\begin{prop}
There exists a surjective group homomorphism from the Clifford group on $n$ qubits 
to the symplectic automorphism group on $\FF_2^{2n}$.
\end{prop}

\begin{exc}
Generalize this to qudits of a prime dimension $d$.
{\em Hint}:
The symplectic transformation $R$ of \ref{SymplecticTransformationGenerators}
needs to be included.
Find the unitary that induces $R$.
The dimension $d$ being prime means that $\FF_d \xrightarrow{ \times a} \FF_d$
for $a \in \FF_d^\times$ is a permutation.
\hfill $\diamond$ \end{exc}

A natural question is then what the kernel of this homomorphism $\varphi$ is.
To find the kernel, suppose $\varphi(U) = \id$, which is to say that
\begin{align}
U P(v) U^\dagger = \eta P(v)
\end{align}
for some phase factor $\eta = \pm 1, \pm i$ that may depend on $v$.
If the action of $U$ on the generators of the Pauli group,
then $U$ is uniquely determined.
The generators of the Pauli group are the Pauli operators $X_i$ and $Z_i$ for qubit $i=1,\ldots,n$.
Since $X_i$ and $Z_i$ are hermitian, we must have $U X_i U^\dagger = \pm X_i$ and
$U Z_i U^\dagger = \pm Z_i$.
Suppose $n=1$, $U X U^\dagger = - X$, and $U Z U^\dagger = Z$.
One solution to these equations is $U = Z$, and we knew that there is a unique solution.
For general $n$, one can assume $U$ is a tensor product of single qubit operators,
and for each factor one finds a Pauli matrix component of $U$.
Thanks to the uniqueness, a solution is the answer.
Therefore, we conclude that the kernel of $\varphi$ is equal to the Pauli group 
(up to an arbitrary phase factor,
which does not alter the conjugation action and hence we ignore).
The Clifford group is generated by Pauli group and three types of elements of \eqref{CliffordGenerators}.
Actually, the Pauli group is generated by the elements of \eqref{CliffordGenerators}.
The phase gate $U_S$ squares to become $Z$.
The Hadamard conjugates it to $X$.
$X$ and $Z$ generates the Pauli group. Hence, we have
\begin{prop}
The three types of Clifford unitary operators of \eqref{CliffordGenerators} generates the full Clifford group.
\end{prop}

Now, suppose we have an additive code defined by a set of stabilizers on $n$ qubits.
We have learned that this stabilizer group $\mathcal S$ 
corresponds to an isotropic subspace $\Sigma$ in the symplectic space $\FF_2^{2n}$.
Let $s = \dim_{\FF_2} \Sigma \le n$.
By \ref{IsotropicSubspaceComplement},
there exists a symplectic transformation that maps $\Sigma$ onto the span of $e_1, \ldots, e_s \in \FF_2^{2n}$.
Since $\varphi$ is surjective, there exists a Clifford unitary that turns the stabilizer group
to $\langle X_1, \ldots, X_s \rangle$ up to signs.
The set of all logical operators after this Clifford unitary is precisely $\langle X_{s+1}, Z_{s+1}, \ldots, X_n, Z_n \rangle$.
This leads to an important conclusion.
Let us say that a set of Pauli operators are {\bf independent} if
the corresponding binary vectors are linearly independent over $\FF_2$.
\begin{thm}
Any stabilizer code defined by $s$ independent stabilizers on $n$ qubits
has code space dimension $2^{n-s}$.
There exist Pauli logical operators that generate the full operator algebra acting on the code space.
\label{ExistenceCompleteLogicalPauli}
\end{thm}
\begin{exc}
Generalize \ref{ExistenceCompleteLogicalPauli}
to qudits $\CC^d$ of prime dimension $d$.
\hfill $\diamond$ \end{exc}
Since the code space dimension is always a power of $2$,
it is convenient to work with the exponent 
\begin{align}
k = n - s,
\end{align}
which is called {\bf the number of encoded or logical qubits}.

\subsection{Cleaning lemma}

We have studied the Pauli group $\mathcal P$,
focusing on its abelianization $\mathcal P / \langle i \rangle$,
which happens to be a vector space,
and the symplectic structure provided by commutation relations.
We converted the code space into the simplest one that we clearly understand,
like the trivial example of additive code defined in \eqref{eq:trivialCodeVectors}.
We did so by considering the largest set of transformations,
the Clifford group,
mapping Pauli operators to Pauli operators.

We are now going to discuss error correction.
For this purpose, the notion of locality is of central importance,
and we should not make a transformation that breaks the notion of locality.
In terms of the vector space associated with the Pauli group,
the locality demands us to use a particular basis.

We say a {\bf region} to mean a subset of qubits.
An operator is said to be {\bf supported} on a region $M$,
if it acts by identity on the complement of $M$.
The {\bf support} of an operator is the minimal region on which the operator is supported.
We learned that given an additive code, there exists a set of logical Pauli operators.
Where are they supported?
The logical operators are the interface of the code space to the external world,
so it is important to know locate them precisely.
For a region $M$, let $\ell_M$ be the largest number of independent logical operators supported on $M$.
Here, the notion of independence is more restrictive than 
what we used to define independent stabilizers in \ref{ExistenceCompleteLogicalPauli}.
Given a stabilizer group $\mathcal S$ and the corresponding isotropic subspace $\Sigma \subset \FF_2^{2n}$,
a set of {\bf independent logical operators} is one that maps to
a linearly independent set in the quotient vector space $\FF_2^{2n} / \Sigma$.

\begin{prop}
\begin{align*}
\ell_M + \ell_{M^c} = 2k.
\end{align*}
\label{LogicalOperatorCounting}
\end{prop}
The case $\ell_M = 0$ is covered in \cite{BravyiTerhal2009no-go},
and more general case is shown in \cite{YoshidaChuang2010Framework}.
The proof below follows \cite{HaahPreskill2012tradeoff}.
\begin{proof}
The set of all Pauli operators supported on $M$
corresponds to a vector space $\FF_2^{2m}$
spanned by $\{ e_i : i \in M \}$, where $m$
is the number of qubits in $M$.
The set of all logical operators on $M$ corresponds to $\Sigma^\perp  \cap \FF_2^{2m}$.
By definition of independent logical operators,
\begin{align*}
\ell_M 
&= \dim_{\FF_2} (\Sigma^\perp \cap \FF_2^{2m}) / ( \Sigma \cap \FF_2^{2m} ) \\
&= \dim_{\FF_2} (\Sigma^\perp \cap \FF_2^{2m}) - \dim_{\FF_2} ( \Sigma \cap \FF_2^{2m} ).
\end{align*}
Let us consider calculating the first term algorithmically.
If we write down the basis vectors for $\Sigma$ in the rows of a matrix $A$,
then $\Sigma^\perp$ amounts to calculating the kernel of the matrix
and transform it with the symplectic matrix $\lambda^{-1}$.
Since we are only interested in the dimension, 
the invertible map $\lambda$ is immaterial.
The restriction ``$\cap \FF_2^{2m}$'' means that we have to find the kernel
with zero components in entries for $M^c$.
This is to say that we calculate the kernel of the submatrix of $A$
obtained by deleting all columns for $M^c$.
The dimension we seek for is then $2m$ minus the rank of this submatrix.
This is a straightforward algorithm, and we translate it back to linear algebra.
The rank of this submatrix is the dimension of $\pi_M(\Sigma)$
where $\pi_M$ is the linear map that sets the components of $M^c$ to zero.
The vectors in $\Sigma_M := \Sigma \cap \FF_2^{2m}$ will remain untouched by $\pi_M$,
those in $\Sigma_{M^c} := \Sigma \cap (\FF_2^{2m})^\perp$ will be annihilated,
and the other vectors will be somehow modified.
Consider the decomposition
\begin{align}
\Sigma = \Sigma_M \oplus \Sigma_{M^c} \oplus \Sigma'
\label{eq:SigmaDecomposition}
\end{align}
where $\Sigma'$ includes whatever remains beyond $\Sigma_M \oplus \Sigma_{M^c}$.
The choice of $\Sigma'$ is not canonical, 
but it is easy to check that $\pi_M |_{\Sigma'}$ is injective,
so $\dim_{\FF_2} \pi_M(\Sigma') = \dim_{\FF_2} \Sigma'$.
Thus,
\begin{align*}
\ell_M 
&= 2m - \dim_{\FF_2} \pi_M (\Sigma) - \dim_{\FF_2} \Sigma_M\\
&= 2m - 2 \dim_{\FF_2} \Sigma_M - \dim_{\FF_2} \Sigma'.
\end{align*}
By symmetric argument,
\begin{align*}
\ell_{M^c} + \ell_M 
&= 2m + 2(n-m) - 2 \dim_{\FF_2} \Sigma_{M} - 2 \dim_{\FF_2} \Sigma_{M^c} - 2 \dim_{\FF_2} \Sigma'\\
&= 2n - 2 \dim_{\FF_2} \Sigma \\
&= 2k
\end{align*}
where the last line is by \ref{ExistenceCompleteLogicalPauli}.
\end{proof}
The proposition has an important corollary:
\begin{thm}[Cleaning Lemma]
If a region $M$ does not support any nontrivial logical operator {\em \bf (correctable)},%
\footnote{
Although in this lecture note we treat the two notions, 
the correctability and the absence of logical operator,
equally, but they are not in general equivalent.
The existence of a correcting map is stronger than the absence of operators
that act nontrivially within the code space.
}
then for any logical operator there exists an equivalent logical operator
supported on the complement region $M^c$.
\label{CleaningLemma}
\end{thm}
In other words, a correctable region can be {\em cleaned} of any logical operators.
\begin{proof}
The assumption is that $\ell_M = 0$. By \ref{LogicalOperatorCounting},
$\ell_{M^c} = 2k$,
which means that every logical operator's action 
can be achieved by a logical operator on $M^c$.
\end{proof}
This simple fact will incur a number of interesting applications below.

\subsection{Error correction criterion}

Imagine one has embedded a state vector into an additive code 
in order to send it over a noisy channel or to store it safely.
A physics qubit may be damaged or lost during the process.
When and how can we recover from the damage?
Let us examine the trivial example in \eqref{eq:trivialCodeVectors} first.
If the error occurs in one of the first $n-k$ qubits, and $\ket 0$ will be mapped to some other state.
Our message is not damaged at all, and the formal restoration is achieved simply by
replacing the first $n-k$ qubits with fresh qubits in the known $\ket 0$ state.
On the contrary, 
if error occurs in one of the last $k$ qubits, then there is no way to recover it;
the damage is permanent.

In fact, the error correction for general additive codes is not too different.
Suppose a physical qubit $i$ in a correctable region is damaged.
By the cleaning lemma \ref{CleaningLemma},
there exists a complete set $\mathcal L$ 
of logical operators supported outside of the damaged qubit $i$.
In particular, every member of $\mathcal L$
commutes with any operator on the damaged qubit $i$.
We learned from the discussion leading to \ref{ExistenceCompleteLogicalPauli}
that there exists a Clifford unitary $U$ that maps the given additive code space
into that of the trivial code of \eqref{eq:trivialCodeVectors}.
$U$ necessarily maps the complete set $\mathcal L$ of logical operators
to a complete set $U \mathcal L U^\dagger$ of logical operators
on the trivial code.
In the trivial code, given an operator $E$ that acts nontrivially on the last $k$ (logical) qubits,
any complete set of logical operators must have one member
that does not commute with $E$. 
Therefore, if the operator $D$ caused the damage on the qubit $i$ in the correctable region,
$U D U^\dagger$ must be supported on the first $n-k$ qubits.
The recovery operation $R$ is then easy,
and the composition $U^\dagger R U$
is our desired error recovery operation.%
\footnote{
The standard notion of the operation
is a quantum channel, a completely positive and trace preserving linear map
on the space of density operators.
Everything we said here can be phrased using channels.
}
Summarizing,
\begin{thm}
For any error that has occurred in a correctable region of an additive code,
there exists a recovery map that corrects it.%
\footnote{
This is not a tautology; we have defined the correctable region because of this result.
}
\end{thm}
This idea goes back to the very first quantum code by Shor~\cite{Shor1996Fault-tolerant}.
A more general criterion was discovered shortly after by Knill and Laflamme~\cite{KnillLaflamme1997Theory}.

The above error correction motivates us to introduce a quantitative attribute to an additive code.
The size of the smallest correctable region is one less than that of
{\em the smallest support of any nontrivial
logical operator}, where the latter is called {\bf code distance} or {\bf minimal distance},%
\footnote{
It is a ``distance'' when we consider the Hamming distance
on the symplectic binary vector space.
This jargon is influenced by classical coding theory.
}
denoted usually by $d$.
So, whenever error occurs on $d-1$ or less qubits, there exists a recovery map.
Obviously, the large $d$ is preferred.

Note that in the above we assumed that we knew the region where the error had occurred.
This is not a very realistic assumption,
and it is necessary to devise a method to locate the error.
This task in general cannot be deterministic,
because 
one has to measure the system in such a way that it does not modify the encoded
quantum state where the measurement outcome is probabilistic in nature,
and two different errors might result in the identical measurement outcomes.
(Often the procedure of locating the error is the hardest step in an 
{\em error correcting algorithm}
since if the locations are known the recovery map is provided by the code itself.)
Therefore, it is important for an error correcting code not only to have large correctable region,
but also to admit a reliable (and efficient) error correcting algorithm.

\begin{exc}
Using the cleaning lemma \ref{CleaningLemma},
show that the code distance $d$ should obey 
$2(d-1) \le n-1$ if $k \ge 1$ for any additive code on $n$ qubits
with $k$ encoded qubits.
(The statement actually follows from {\em quantum Singleton bound}
$2(d-1) \le n-k$~\cite{KnillLaflamme1997Theory},
but the weaker version can be derived from the cleaning lemma.)
\hfill $\diamond$ \end{exc}

\subsection{Entanglement spectrum}

A bipartite state  vector $\ket{\psi_{AB}}$ can always be written as
\begin{align}
\ket{\psi_{AB}} = \sum_i \sqrt{p_i} \ket{\phi_A^{(i)}} \ket{\phi_B^{(i)}}
\end{align}
for positive numbers $p_i$ and some 
orthonormal vectors $\ket{\phi_A^{(i)}}$ and $\ket{\phi_B^{(i)}}$
where $\sum_i p_i = 1$ by the normalization $\braket{\psi_{AB}}{\psi_{AB}} = 1$.
The numbers $p_i$ are called the Schmidt coefficients of the state $\ket{\psi_{AB}}$,
and also called the {\bf entanglement spectrum}.
The entanglement spectrum is actually a complete set of {\em invariants}
under unitary transformations on either partition;
it is invariant under unitaries, and conversely 
the entanglement spectrum determines the pure state
up to unitary transformations on each partition.

A code state vector is defined on a Hilbert space consisting of $n$ tensor factors.
Every choice of a subset of qubits (region) defines a bipartition,
and one can ask what the entanglement spectrum of a code state is.
In general, the entanglement spectrum depends on the particular code vector,
but if the code vector is an eigenvector of a maximal set of commuting
logical Pauli operators, then the entanglement spectrum turns out to be very simple.
The commuting set of logical operators can be regarded as Pauli stabilizers,
so the state vector is uniquely determined by a set of $n$ commuting Pauli stabilizers.
The eigenvalues of the stabilizers need not be all $+1$.
\begin{prop}
Let $\ket \psi \in (\CC^2)^{\otimes n}$ be a nonzero common eigenvector of 
$n$ independent Pauli operators $P_i$.
With respect to any bipartition, the entanglement spectrum is independent of the eigenvalues.
In other words, the entanglement spectrum is determined by
the binary vectors of the Pauli operators $P_i$.
\label{EntSpecIndependentOfEigenvalues}
\end{prop}
\begin{proof}
First the $n$ independent Pauli operator must all commute with one another;
otherwise, if $P$ and $Q$ are any two anti-commuting Pauli operators
with eigenvalues $p,q$, respectively, then
\begin{align}
pq \ket \psi = p Q \ket \psi = Q p \ket \psi = Q P \ket \psi = - P Q \ket \psi = -q P \ket \psi = -pq \ket \psi
\label{eq:antiCommutingStabilizers}
\end{align}
so $\ket \psi = 0$.
By \ref{IsotropicSubspaceComplement},
there exists a Clifford unitary $U$ that conjugates the stabilizer group $\mathcal S$ of the given commuting Pauli operators
to that of $X_i$ ($i=1,\ldots,n$).
The state $U \ket \psi$ is a common eigenstate of $X_i$ with some eigenvalues.
Two states with different eigenvalues $\pm 1$ of $X_i$ can be mapped to each other by some $Z_i$.
Since $U$ is Clifford, $U Z_i U^\dagger$ is also a Pauli operator,
and maps $\ket \psi$ to another state that is a common eigenvector of $\mathcal S$.
Since a Pauli operator $U Z_i U^\dagger$ is a tensor product unitary operator,
it cannot change the entanglement spectrum with respect to any bipartition.
\end{proof}

\begin{thm}
Given any bipartition $M \sqcup M^c$ of $n$ qubits,
and a nonzero common eigenvector $\ket \psi$ of $n$
independent Pauli operators,
there exists a tensor product Clifford operator
that transforms $\ket \psi$ into Bell pairs.
In particular, the entanglement spectrum is flat.
The number of nonzero Schmidt coefficients is equal to $2^s$
where
\begin{align}
 s = |M| - \dim_{\FF_2} \Sigma_M .
\label{eq:entropy}
\end{align}
Here, $|M|$ is the number of qubits in $M$ and
$\dim_{\FF_2} \Sigma_M$ is the number of independent stabilizers supported on $M$.
\label{EntanglementSpectrum}
\end{thm}
The last formula goes back at least to \cite{FattalCubittYamamotoEtAl2004Entanglement}.
(See also \cite{HammaIonicioiuZanardi2005}.)
It can also be easily derived from \cite{KlappeneckerRoetteler2002stabilizer};
see \cite{LindenMatusRuskaiEtAl2013Quantum} and references therein.
\begin{proof}
From \ref{EntSpecIndependentOfEigenvalues},
we may assume that $n$ independent Pauli operators $P_i$ are commuting and
have eigenvalue $+1$.
Order the qubits so that those in $M$ are the first $m = |M|$ qubits.

Suppose some $P_i$ is supported on $M$.
The single operator $P_i$ defines an isotropic $\FF_2$-subspace of $\FF_2^M$,
and by \ref{IsotropicSubspaceComplement},
we see that there exists a Clifford unitary $U$ supported on $M$
such that $P_i$ is mapped to $X_1$.
The eigenstate of $X_1$ is always of form $\ket{+} \otimes \star $,
i.e., the first qubit becomes disentangled by $U$.
This Clifford does not affect the entanglement spectrum between $M$ and $M^c$,
and hence we can remove the first qubit and we are left with one less qubit
and one less stabilizer $X_1$.
The formula \eqref{eq:entropy} retains its form since $|M|$ is reduced by 1,
and simultaneously $\dim_{\FF_2} \Sigma_M$ is reduced by 1.
Therefore, without loss of generality we may assume no stabilizer $P_i$
is supported only on $M$ or $M^c$.

Consider the matrix $\Sigma$ that has the binary vectors for $P_i$ in the columns.
The rows $1,\ldots, m$ and $n+1, \ldots, n+m$ are for $M$.

We first claim that $n = 2m$.
To show this, consider the decomposition of $\Sigma$ in \eqref{eq:SigmaDecomposition}.
There we showed that the restriction map $\pi_M$ that sets the component of $M^c$ to zero
is injective on $\Sigma'$. By our assumption, $\Sigma = \Sigma'$,
and hence $\pi_M$ is injective.
This means that the $2m \times n$ submatrix $\Sigma^{(M)}$ of $\Sigma$ consisting of rows for $M$
has the same rank as the full matrix $\Sigma$.
The rank of $\Sigma$ is $n$ by assumption.
This demands that $2m \ge n$.
Repeating the argument for $M^c$ in place of $M$, we have $2(n-m) \ge n$.
This proves $n = 2m$.

The column operation of $\Sigma$ is nothing but a different choice of
independent stabilizers of $\ket \psi$.
Run column operations on $\Sigma$ such that the $2m \times 2m$ submatrix $\Sigma^{(M)}$
becomes the identity matrix.
\begin{align*}
 \Sigma' =
\begin{pmatrix}
1                  & 0_{1 \times (m-1)} &  &  \\
0_{(m-1) \times 1} & \id_{(m-1)}        &  &  \\
\star              & \star & \star & \star \\
\star & \star        & \star & \star \\
\hline
 &  & 1 & 0_{1 \times (m-1)} \\
 &  & 0_{(m-1) \times 1} & \id_{(m-1)}\\
 \star              & \star & \star & \star \\
\star & \star        & \star & \star
\end{pmatrix}
\end{align*}
Next, we apply symplectic transformation on $M$ or $M^c$ without changing the entanglement spectrum,
so that the transformed matrix of $\Sigma$ will be simple.
The $(4m-2) \times 1$ submatrix in the first column corresponding to $M^c$ entries
cannot be zero, since otherwise the first column will be supported on $M$.
As in the proof of \ref{IsotropicSubspaceComplement},
we use the elementary symplectic transformations on $M^c$ to obtain
\begin{align*}
 \Sigma'' =
\begin{pmatrix}
1                  & 0_{1 \times (m-1)} &  &  \\
0_{(m-1) \times 1} & \id_{(m-1)}        &  &  \\
1              & \star & \star & \star \\
0_{(m-1) \times 1} & \star        & \star & \star \\
\hline
 &  & 1 & 0_{1 \times (m-1)} \\
 &  & 0_{(m-1) \times 1} & \id_{(m-1)}\\
0              & E_{1 \times (m-1)} & f & G_{1 \times (m-1)} \\
0_{(m-1)\times 1} & \star        & \star & \star
\end{pmatrix}.
\end{align*}
We now employ the equation $(\Sigma'')^T \lambda (\Sigma'') = 0$.
Due to the first column, the equation enforces $E = 0$, $f=1$, $G = 0$.
\begin{align*}
 \Sigma''' =
\begin{pmatrix}
1                  & 0_{1 \times (m-1)} &  &  \\
0_{(m-1) \times 1} & \id_{(m-1)}        &  &  \\
1              & \star & \star & \star \\
0_{(m-1) \times 1} & \star        & \star & \star \\
\hline
 &  & 1 & 0_{1 \times (m-1)} \\
 &  & 0_{(m-1) \times 1} & \id_{(m-1)}\\
0              & 0_{1 \times (m-1)} & 1 & 0_{1 \times (m-1)} \\
0_{(m-1)\times 1} & \star        & \star & \star
\end{pmatrix}.
\end{align*}
Again by $C$ operations and $S$ operations on $M^c$, we obtain
\begin{align*}
 \Sigma'''' =
\begin{pmatrix}
1                  & 0_{1 \times (m-1)} &  &  \\
0_{(m-1) \times 1} & \id_{(m-1)}        &  &  \\
1              & \star & 0 & \star \\
0_{(m-1) \times 1} & \star        & 0_{(m-1) \times 1} & \star \\
\hline
 &  & 1 & 0_{1 \times (m-1)} \\
 &  & 0_{(m-1) \times 1} & \id_{(m-1)}\\
0              & 0_{1 \times (m-1)} & 1 & 0_{1 \times (m-1)} \\
0_{(m-1)\times 1} & \star        & 0_{(m-1) \times 1} & \star
\end{pmatrix}.
\end{align*}
We see that the first column and $(m+1)$-st column
is block diagonal with the rest.
These isolated columns represent the state that is stabilized by $X_1 X_m$ and $Z_1 Z_m$.
There is a unique such state $(\ket{00} + \ket{11})/\sqrt{2}$, the {\bf Bell pair}.
The entanglement spectrum is $\{1/2,1/2\}$, and there are two nonzero Schmidt coefficients.
The formula \eqref{eq:entropy} is clearly valid.
\end{proof}

For generalizations to prime $d$-dimensional qudits,
the entry $f$ in $\Sigma''$ should be $-1$.
The number of nonzero Schmidt coefficients is a power of $d$.

\begin{cor}
For any state in an additive code space,
the entanglement spectrum of any correctable region $M$ is flat.
The entanglement entropy of $M$ is given by \eqref{eq:entropy}.
\label{EntanglementSpectrumCorrectableRegion}
\end{cor}
\begin{proof}
By the cleaning lemma \ref{CleaningLemma},
any logical operator can be pushed away from $M$.
This implies that any observable supported on $M$ has the same expectation value 
regardless of the actual encoded state.
In other words, the reduced density matrix for $M$ is independent of the code state.
Therefore, we can conveniently choose the code state to be an eigenstate of
a maximal set of commuting logical operators,
and we can use \ref{EntanglementSpectrum}.
The entanglement entropy is $- \sum_i p_i \log p_i$,
which is equal to $\log$ of the number of nonzero Schmidt coefficients
because they are the same.
\end{proof}

\section{Geometric locality}
\label{sec:GeometricLocality}

From now on, we consider additive codes whose stabilizer group is generated by
Pauli operators supported on small balls in some metric space.
We will call an upper bound $w$ on the ball's diameter 
as the the {\bf geometric locality} of the code.

The error correcting capability of the additive codes
is, crudely speaking, due to the entanglement structure of the code states.
If the code distance is larger than, say $3$,
any pair of qubits have the same reduced density matrix for any code state.
In other words, the complete data for all pairs of qubits cannot determine the code state at all;
the entanglement made the global information hidden from local degrees of freedom.
In \ref{OneDCodeDistanceBound} below we will see that this interesting property
is severely restricted by the geometric locality.

This signals implications in many-body physics.
Physical degrees of freedom, qubits, interact with each other
whose strength is only strong for those that are nearby.
We cannot in general fully determine the microscopic interaction,
and, even if we did, it is almost always hopeless to calculate
consequences of the interaction exactly.
Instead, physicists model the system and identify important aspects,
which can be compared with experiments.
The geometrically local additive codes provides a class of 
physically relevant (i.e., local) models which we can analyze relatively easily.
Since the additive codes is designed to produce highly entangled states,
the intuition gained from this class of models
will be valuable to enhance our understanding of physical systems
where entanglement is presumably essential.

Another practical motivation to study geometrically local additive codes
is to use them as ``firmware'' in quantum information processing architectures.
A raw physical qubit is likely to be noisy,
so it is expected that an error correction layer will be added on top of
a system of physical qubits.
The error correction scheme would be implemented easier
if the necessary operation is on a small local cluster of qubits.
One of the most important operations in any error correcting code implementation
is to check whether a state is in the code space.
In geometrically local codes,
this membership test can be done on local clusters of qubits.

The first result under the geometric locality is the following.
\begin{lem}[Union Lemma~\cite{BravyiTerhal2009no-go,HaahPreskill2012tradeoff}]
Let $M$ and $N$ be correctable regions for a geometrically local additive code of locality $w$.
If $M$ and $N$ are separated by distance $>w$,
then $M \cup N$ is also correctable.
\label{UnionLemma}
\end{lem}
\begin{proof}
We have to show that any logical operator $O$ supported on $M \cup N$ is trivial.
Since $O$ is a tensor product operator,
we can write $O = O_M \otimes O_N$ where $O_M$ is supported on $M$
and $O_N$ on $N$.
The locality implies that there is no stabilizer group generator
that acts nontrivially on both $M$ and $N$.
By \ref{LogicalOperatorsCommutant}, each stabilizer group generator must commute
with each tensor factor $O_M$ and $O_N$.
This means that each of $O_M$ and $O_N$ is a logical operator.
Since each region is correctable, $O_M$ is trivial, so is $O_N$.
The product $O$ is also trivial.
\end{proof}

Below we give important applications of the cleaning lemma \ref{CleaningLemma}
and the union lemma \ref{UnionLemma}.

\subsection{Code distance is bounded in one dimension}
Suppose we have an array of $n$ qubits along a line.
We consider an additive code whose stabilizer group is
generated by Pauli operators supported on intervals of length at most $w$.
How large can the code distance $d$ be?
\begin{thm}
$d \le 3w$ if the additive code encodes at least one qubit.
\label{OneDCodeDistanceBound}
\end{thm}
The result in Ref.~\cite{BravyiTerhal2009no-go} is essentially this,
though they have phrased it primarily with two-dimensional systems.
\begin{proof}[Proof 1~\cite{HaahPreskill2012tradeoff}]
Suppose on the contrary that $d > 3w$.
Then any interval of length $\le 3w$ is correctable.
For integers $r \ge 0$,
let $M_r = [-2r w, 3w+2rw]$ be intervals.
By the cleaning lemma,
$M_0$ is correctable.

We claim by induction that $M_r$ for any $r \ge 0$ is correctable.
$M_{r+1} \setminus M_{r}$ consists of two intervals of length $2w$
separated by distance $>w$.
By the union lemma, the union of the two intervals is correctable.
If $O$ is any logical operator,
by the cleaning lemma an equivalent logical operator $O'$ exists on the complement of 
$M_{r+1} \setminus M_{r}$, which is $M_r \sqcup M_{r+1}^c$.
Since $O'$ is a tensor product operator,
we may consider its tensor factor $O'_{M_r}$ on $M_r$,
which by itself must be logical due to the locality.
Since $M_r$ is correctable by induction hypothesis,
$O'_{M_r}$ is a stabilizer, and the tensor factor $O'_{M_{r+1}^c}$
is equivalent to $O$.
Since $O$ was arbitrary, this implies that on $M_{r+1}^c$
a complete set of logical operators can be found, i.e., $\ell_{M_{r+1}^c} = 2k$,
and by \ref{LogicalOperatorCounting} we conclude that $M_{r+1}$ is correctable.
This completes the induction.

For a sufficiently large $r$, $M_r$ includes all $n$ qubits,
and the whole system is correctable.
This contradicts the assumption that there exists a nontrivial logical operator.
\end{proof}
\begin{proof}[Proof 2~\cite{BravyiTerhal2009no-go}]
For integer $r \in \mathbb Z$, let $N_{r} = [2rw, 2(r+1)w]$.
Consider $N_\text{even} = \bigcup_{r \in 2\mathbb Z} N_r$ and $N_\text{odd} = \bigcup_{r \in 2 \mathbb Z +1} N_r$.
If $d > 3w$, then $N_\text{even}$ is correctable by the union lemma.
A nontrivial logical operator can then be found in $N_\text{odd}$ by the cleaning lemma,
which is a contradiction since $N_\text{odd}$ is correctable as well by the union lemma.
\end{proof}

\begin{exc}
Sharpen the statement of \ref{OneDCodeDistanceBound}
using the argument in Proof 2.
\hfill $\diamond$ \end{exc}

\begin{exc}
Apply the conclusion of \ref{OneDCodeDistanceBound}
for codes in higher dimensional lattices,
to show that there always exists a nontrivial logical operator
supported on a thin slab~\cite{BravyiTerhal2009no-go, HaahPreskill2012tradeoff}.
\hfill $\diamond$ \end{exc}

\section{Translation invariance}

We specialize to the additive codes that obeys translation invariance.
The material of this section and the next is 
mostly from Ref.~\cite{Haah2012PauliModule}.
An exception is the proof of \ref{CSS2dClassification}.
Specifically, we consider an array of qubits on a lattice.
The lattice is modeled by the additive group $\ZZ^D$,
where $D$ is called the {\bf spatial dimension}.%
\footnote{
This is yet another ``dimension,''
which is different from the Hilbert space dimension,
or the binary vector space dimension.
}
The lattice is a collection of {\bf sites}, the elements of $\ZZ^D$.
We assume that a finite number $q$ of qubits are located on each site.
A local operator on this infinite array of qubits
is an operator supported on a ball of finite size.
The set of all local operators naturally admits an action from $\ZZ^D$
by permuting (translating) the tensor factors.
A {\bf translation invariant code} is one for which the generating set
of the stabilizer group consists of local Pauli operators and 
is translation-invariant.

This generalizes our previous discussion,
since if we set $D=0$ the notion of translation becomes vacuous.
In other words, our study of finite dimensional symplectic vector spaces
and automorphism groups was zero-dimensional.

For a positive $D$, 
one might be worried about the infinitely large lattice and the corresponding
infinite dimensional Hilbert space.
We are not going to discuss about this infinity.
Instead, we contend ourselves by implicitly considering a family of codes
by factoring out sublattices (subgroup) of $\ZZ^D$ of finite subgroup indices.
This index, the order of the quotient group,
is more often called {\bf system size}.
The procedure of factoring out the subgroup of $\ZZ^D$
is to impose {\bf periodic boundary conditions} on the lattice.

The translation invariance allows us unambiguously define the stabilizer group
over the family of finite systems.
It should be noted that the absence of $-1$ from the stabilizer group (See \ref{NoMinusOne})
is not always guaranteed.
For instance, in $D=1$ consider the stabilizer group generated by $-Z_i Z_{i+1}$.
On the ring of odd length, a product of these generators is equal to $-1$.
We intentionally avoid such a situation, and focus on those where $-1$ does not appear.%
\footnote{
It can be shown that for any finite abelian group of Pauli operators,
there exists a finite abelian group of Pauli operators that does not contain $-1$,
but has the same binary vector representation.
}

Concretely, the translation-invariance helps us to deal with the formal infinite dimensionality.
The bit string that encodes Pauli operator up to a phase factor
is really a specification of the support of the Pauli operator.
For a Pauli operator on lattice, we have the coordinate system $\ZZ^D$,
so a list of integer $D$-tuples is all we need.
Furthermore,
a Pauli operator is a finite product of single qubit operator,
which we can write as (assuming $D=3$)
\begin{align*}
 P( a,b,c ; e_i )
\end{align*}
where $e_i$ denotes qubit $i$ within the site (a unit cell) $(a,b,c) \in \ZZ^3$.
By convention, $e_i$ for $1 \le i \le q$ means the Pauli $X$,
and $e_{q+i}$ for $1 \le i \le q$ means the Pauli $Z$.
The Pauli $Y$ is represented as $e_{i} + e_{i+q}$.
Let us employ formal variables $x,y,z$ associated with the generators of $\ZZ^3$,
and write
\begin{align*}
 (a,b,c) \Leftrightarrow x^a y^b z^c .
\end{align*}
In previous section, for multiplication of Pauli operators,
we added the corresponding bit strings over $\FF_2$.
Similarly we write, for example,
\begin{align}
 P(1,2,3; e_1 ) P( -1,2,4; e_{q+2} ) \Leftrightarrow x y^2 z^3 e_1 + x^{-1} y^2 z^4 e_{q+2} 
\end{align}
for a general Pauli operator.
These observations can be summarized as
\begin{prop}
The following two groups are isomorphic.
\begin{itemize}
\item The multiplicative group of {\em finitely supported Pauli operators} modulo phase factors
on the lattice of dimension $D$ with $q$ qubits per site
\item The additive group of all {\em Laurent polynomial column vectors}%
\footnote{
We have overloaded the term ``vectors.''
It is a $2q \times 1$ matrix over the Laurent polynomial ring, but this is too wordy.
As a vector over the base field $\FF_2$, the number of components is infinite.
Here, we just say column vectors to mean anything that has several components
arranged in a column.
}
 in $D$ formal variables of length $2q$
\end{itemize}
\end{prop}

Recall that a vector space is an additive group with an action from a field.
Likewise, a {\bf module} $M$ is an additive group with an action (``$\cdot$'') 
from a more general ring $R$:
$r \cdot m \in M$ for any $r \in R$ and $m \in M$.
The distributive law is assumed by the same formula:
$(r_1 + r_2) \cdot m = r_1 \cdot m + r_2 \cdot m$
and
$r \cdot ( m_1 + m_2 ) = r \cdot m_1 + r \cdot m_2$.
We will generally omit the dot (``$\cdot$'') of the ring action.

The translation group acts on the set of Pauli operators,
and hence on the set of Laurent polynomial column vectors.
The translation along $x$-direction by a unit distance increases the $x$-coordinate by $1$,
which is equivalent to multiplication by $x$ on the Laurent polynomial column vector.
This action naturally defines an action from the group ring
$R = \FF_2[ \mathbb Z^D] \cong \FF_2[ x_1 ^{\pm 1} , \ldots, x_D^{\pm}]$.
Therefore, we speak of {\bf Pauli module} over the the translation group algebra $R$
for the additive group of all Laurent polynomial column vectors of length $2q$.

By the definition of the translation invariant codes,
the stabilizer group is closed under the translation group action.
This is to say that the set of Laurent polynomial column vectors corresponding
to the stabilizer group
is an $R$-submodule, which we call {\bf stabilizer module}, of the Pauli module.

In the zero-dimensional study, the symplectic form emerged from the commutation relation.
We interchanged the vector components corresponding to $X$ and $Z$ parts,
and using the dot product we counted the number of overlap.
The interchange of $X$ and $Z$ components is readily defined
for the Pauli module elements,
but dot product is not immediately applicable.
For Laurent polynomial column vectors, what we want is
to count the number of overlapping terms.
A trick is to consider the following antipode, denoted by the bar over the element.

\subsection{Symplectic form}

The {\bf antipode map}%
\footnote{
It is the antipode map of the group algebra $R = \FF[\ZZ^D]$
taken as a Hopf algebra.
}
is an involutive%
\footnote{
The inverse is itself.
}
$\FF$-linear map from $R$ to $R$
defined by
\begin{align}
  x_1^{a_1} \cdots x_D^{a_D} \mapsto \overline{x_1^{a_1} \cdots x_D^{a_D}} := x_1^{-a_1} \cdots x_D^{-a_D} .
\end{align}
It is then clear that for any $f, g \in R$,
\begin{align}
 \text{Sum of the coefficients of overlapping terms of $f$ and $g$} = \text{Coefficient of 1 in } \bar f g .
\end{align}
The latter quantity actually defines a $\FF$-bilinear form on $R$ (valued in $\FF$),
and naturally generalizes to $R^n$, which we call {\bf dot product}:
\begin{align}
\text{For any } v, w \in R^n, \text{ define } (v \cdot w) = 
\left(
\text{Coefficient of 1 in } \sum_{i=1}^n \bar v_i w_i  
\right)
\in \FF.
\end{align}
Clearly, this dot product generalizes the zero-dimensional ($D=0$) dot product.
Note that
\begin{align}
 (g v \cdot w ) = ( v \cdot \bar g w) \text{ for any } g \in R, ~v,w \in R^n .
 \label{eq:DotProduct2}
\end{align}
The generalization of \eqref{eq:symp}, the {\bf symplectic form}, on $R^{2n}$ is given by
\begin{align}
 \lambda(v,w) = (v \cdot \lambda_n w)
\end{align}
where $\lambda_n$ is the $2n \times 2n$ symplectic matrix \eqref{eq:symplecticMatrix}.
We used the same symbol $\lambda$ as it is a generalization of the zero-dimensional case.
With respect to this symplectic form,
we continue to say that an $R$-submodule $\Sigma$ of $R^n$ is {\bf isotropic} if
\begin{align}
 \forall v,w \in \Sigma, \quad \lambda( v,w) = (v \cdot \lambda_n w) = 0.
\end{align}
\begin{prop}
A submodule $\Sigma$ of $R^{2n}$ generated by the columns of a Laurent polynomial matrix $\sigma$
is isotropic if and only if
\begin{align}
 \bar \sigma^T \lambda_n \sigma = 0
\end{align}
as a matrix. We will denote $\bar \sigma^T$ by $\sigma^\dagger$.
\label{IsotropyCondition}
\end{prop}
\begin{proof}
Say $\sigma$ is $2n \times t$.
($\Leftarrow$)
An arbitrary element of $\Sigma$ is $\sigma h$ for some column vector $h \in R^t$.
Thus, 
\begin{align*}
 \lambda( \sigma h, \sigma h' ) = \text{Coeff. of 1 in }( h^\dagger \sigma^\dagger \lambda_n \sigma h' ) = 0 .
\end{align*}
($\Rightarrow$)
We must have
\begin{align*}
\text{Coeff. of 1 in }( e_i^\dagger \sigma^\dagger \lambda_n \sigma g^{-1} e_j ) = 0 
\end{align*}
for any basis unit vector $e_i$, $e_j$, and any monomial $g$.
$e_i$ and $e_j$ selects $(i,j)$-entry of the matrix $\sigma^\dagger \lambda_n \sigma$.
The coefficient of $1$ in a Laurent polynomial $f$ 
that is shifted by $g^{-1}$ is exactly the coefficient of $g$ in $f$.
Since $g$ addresses any term,
this implies that the $(i,j)$-entry is identically zero.
\end{proof}

\begin{rem}
The dot product is a non-degenerate symmetric bilinear form on $R$.
That is, if $(v \cdot w ) = 0$ for all $w \in R$, then we must have $v = 0$.
This means that the $\FF$-linear map $v \mapsto f_v \in R^*$, 
where $f_v: w \mapsto (v \cdot w)$ is a linear functional on $R$,
is injective.
This association is not surjective whenever $D > 0$,
but becomes surjective if we consider ``finite systems''
by imposing conditions such as $x_i^{L_i} = 1$
since $R$ becomes finite dimensional $\FF$-vector space.

The dual space $R^*$ can be endowed with an $R$-module structure
by defining $r f( \star ) = f( \bar r \star )$ for any $r \in R$ and $f \in R^*$.
If $f = (v \cdot \star)$, then $r f = f( \bar r \star ) = (v \cdot \bar r \star) = (rv \cdot \star)$ 
by \eqref{eq:DotProduct2}.

Given a $R$-module map $\varphi : R^n \to R^m$, i.e., an $m \times n$ matrix with entries in $R$,
we can consider its dual $\varphi^* : (R^m)^* \to (R^n)^*$ by the rule
$(\varphi^* f)(\star) = f( \varphi(\star))$.
If the association $v \mapsto (v \cdot \star)$ is bijective,
then we can consider 
\[
R^m \xrightarrow{\cong} (R^m)^* \xrightarrow{\varphi^*} (R^n)^* \xrightarrow{\cong} R^n
\]
and ask the matrix representation of this $R$-module map.
The answer is simply $\varphi^\dagger$.
This is in fact how the transpose for real matrices and the hermitian conjugate for complex matrices
are defined with respect to the usual inner product.
\hfill $\diamond$ \end{rem}

\begin{rem}
The base ring $R = \FF[\ZZ^D]$ can be obtained from a polynomial ring in $D$ variables
by inverting a single element $x_1 \cdots x_D$.
In particular, it is Noetherian.
Hence, every submodule of $R^n$ is finitely generated,
which means that there always exists a (rectangular) finite matrix $\sigma$
whose columns generate the given submodule over $R$.
Almost all statements below require that the stabilizer module
is finitely generated,
which is guaranteed by the finite dimensional lattice $\ZZ^D$.
For thorough treatment of Noetherian rings and modules,
consult Atiyah-MacDonald~\cite{AtiyahMacDonald}.
\hfill $\diamond$ \end{rem}

\subsection{Stabilizer and Excitation map}

When there are $q$ qubits per site on the lattice $\ZZ^D$,
the matrix equation of \ref{IsotropyCondition} reads
\begin{align}
\sigma^\dagger \lambda_q \sigma = 0.
\end{align}

Note that this is a genuine generalization of the zero-dimensional case $D=0$.
The number $q$ is the total number of qubits in the system,
and $\dagger$ reduces to the usual transpose.
The matrix $\sigma$ is usually called a {\bf generating matrix} of the code on $q$ qubits.
In our convention, the columns of $\sigma$ are the stabilizer group generators.%
\footnote{
Many other references make convention where rows of a binary matrix represent
stabilizer group generators.
}
We call $\sigma$ a {\bf stabilizer map}
in order to emphasize the importance of the image of this $R$-linear map 
rather than the particular matrix representation.

It will prove useful to think of the $t \times 2q$ matrix 
\begin{align}
\epsilon = \sigma^\dagger \lambda_q
\end{align}
separately from $\sigma$ and define a chain complex
\begin{align}
(E = R^t) \xleftarrow{\epsilon} (P=R^{2q}) \xleftarrow{\sigma} (G = R^t)
\end{align}
of length $2$.
(A {\bf chain complex} is an array of maps such that the composition of any consecutive maps is zero.
It has nothing to do with complex numbers.)

Previously in Section~\ref{sec:GeometricLocality},
we briefly noted that the geometrically local additive codes model
physical systems.
This is because local codes define a Hamiltonian
\begin{align}
H = - \sum_g P_g
\label{eq:Hamiltonian}
\end{align}
where $P_g$ are the hermitian local stabilizer generators.
The generators may be redundant.
The Hamiltonian actually depends on the particular choice
of generators, and is not uniquely determined by the code space.
The lowest energy eigenspace ({\em ground} space) of $H$, however, is by definition 
the code space regardless of the choice of $P_g$.

What about higher energy states?
These {\em excited} eigenspace decomposes into the eigenspaces of individual generators $P_g$;
the ground space has eigenvalue $P_g = +1$.
For each eigenspaces, 
we may visualize the distribution of the eigenvalues of $P_g$'s in the lattice,
as $P_g$ is supported on a small region of the lattice.
These eigenvalues are precisely what we can measure
without worrying about corrupting the encoded states in the ground (code!) space.
Any eigenvalue $-1$ of $P_g$ is called a {\bf defect} or {\bf excitation}.
\begin{exc}
Show that for any common eigenspace of $P_g$'s, 
there exists a Pauli operator that maps the ground space onto it.
\hfill $\diamond$ \end{exc}

In the translation invariant codes, the Hamiltonian can also be chosen to be translation invariant,
and the choice of the Hamiltonian is conveniently expressed by
a particular matrix $\sigma$ of the stabilizer map.
The columns represent Pauli operators supported near the origin,
and all the other terms in the Hamiltonian is obtained by translations.
In this case, the excitations are less sensitive to the choice of the matrix $\sigma$.
Consider a state $\ket{\psi_e}$ that has an excitation $P_i = -1$ around a site $i$ with respect to the $\sigma$.
If we choose another Hamiltonian, $- \sum_{g'} P'_{g'}$, for the same code space, represented by $\sigma'$,
then the operator $P_i$ as a Laurent polynomial column vector $v$ should be given by
some $R$-linear combination of the columns of $\sigma'$:
\begin{align}
v = \sum_j a_j \sigma'_j, \quad a_j \in R.
\end{align}
Since $a_j$ consists of finitely many terms,
this means that $P_i$ is a product of some finitely many local terms $P'_{g'}$ around $i$,
among which there must be an operator $P'_{g'}$ that has eigenvalue $-1$ on $\ket{ \psi_e}$.
The location of the excitation $P'_{g'} = -1$ is not too different from that of $P_i = -1$.

The matrix $\epsilon$ yields a convenient way to determine the locations of the defects (excitations)
when a Pauli operator acts on the ground state.
Let $\ket{\psi}$ be a ground state (code state), and let $P$ be arbitrary Pauli operator that is finitely supported.
To locate the excitations of $\ket{\psi'} = P \ket \psi$, we consider
\begin{align}
P_g \ket{\psi'} = P_g P \ket \psi = \pm P P_g \ket \psi = \pm P \ket \psi = \pm \ket{\psi'}
\end{align}
where the sign $\pm$ is determined by the commutation relation between $P_g$ and $P$.
We know how to express this sign by the symplectic form.
The Hamiltonian term $P_g$ is a translation by $g \in \ZZ^D$ 
of one of the columns of $\sigma$, say $i$-th column $\sigma_i$.
The Pauli operator $P$ is expressed by some column vector $v$.
The anti-commutation happens precisely when
\begin{align}
1 = (g\sigma_i \cdot \lambda_q v)  = ( g \cdot \sigma_i^\dagger \lambda_q v).
\end{align}
If we vary $g \in \ZZ^D$, 
then we will collect all positions of the excitation associated with $\sigma_i$,
and may convert these position data into a Laurent polynomial.
But, this Laurent polynomial is exactly $\sigma_i^\dagger \lambda_q v$.
If we collect these polynomials, one for each column $i$ of $\sigma$,
we obtain a column vector of Laurent polynomials of length $t$,
which is equal to $\sigma^\dagger \lambda_q v$.
This shows how $\epsilon = \sigma^\dagger \lambda_q$
defines a map from the Pauli operators to the excitations.
The map $\epsilon$ is the {\bf excitation map}.

\begin{exc}
Recall \ref{LogicalOperatorsCommutant}.
Assume periodic boundary conditions for the lattice
to work with finite systems,
and show that a Pauli operator is logical if and only if 
its Laurent polynomial representation $v$ satisfies $\epsilon(v) = 0$.
Show also that it is nontrivial if $v \notin \im \sigma$.
\label{LogicalOperatorExcitationMap}
\hfill $\diamond$ \end{exc}

\subsection{Symplectic/Locality-preserving Clifford transformations}
\label{sec:LatticeSympTrans}

In the zero-dimensional study above,
we identified symplectic transformations that preserves the symplectic form
in the abstract vector space, and found a generating set for the symplectic group.
We then showed that, on the $\FF_2$-symplectic space 
derived from the commutation relation among Pauli operators on qubits,
the generators of the symplectic group
are induced by Clifford unitary transformations on qubits.
Here we parallel the discussion with the translation invariance.
Recall that $q$ is the number of qubits per site, 
and $R = \FF_2 [x_1^\pm, \ldots, x_D^\pm]$ is our base ring, the translation group algebra.

A $2q \times 2q$ matrix $T$ on $R^{2q}$ is symplectic if it satisfies the matrix equation
\begin{align}
 T^\dagger \lambda_q T = \lambda_q .
\label{SymplecticCondition}
\end{align}
Restricting the entries of $T$ to be in $\FF_2$ we recover the symplectic group in the zero-dimensional case.
This subgroup of the symplectic group is induced by the application of
the Hadamard, the phase, and the controlled-not gate on every unit cell, uniformly over the lattice.

There are other symplectic transformations,
of which we enumerate a few.
When $q=1$,
for any monomial $g=x_1^{a_1} \cdots x_D^{a_D}$, the matrix $ \begin{pmatrix} g & 0 \\ 0 & g  \end{pmatrix}$
is symplectic. Interpreting in terms of action on qubits,
this amounts to translating qubits into $(a_1,\ldots,a_D)$-direction.
This certainly maps a Pauli operator to a Pauli operator and preserves the size of the support;
it is a locality-preserving Clifford transformation.

Assuming $q=1$ still,
we consider transformation of form $S_f = \begin{pmatrix} 1 & 0 \\ f & 1 \end{pmatrix}$.
Plugging it into \eqref{SymplecticCondition}, we obtain an equation $f = \bar f$.
The matrix $S$ of \eqref{ElementaryZeroDSymplecticTransformation} corresponds to 
$S_f$ for $f \in \FF_2 \subseteq R$.
$S_f$ with a general $f$ can be split into a product of 
finitely many $S_{m + \bar m}$ for a monomial $m$ 
and at most one $S_a$ for $a \in \FF_2$.
The transformation $S_{m + \bar m}$ is new arising from the translation structure,
and is induced by a {\em controlled}-$Z$ gate
\begin{align}
 U_{CZ} = \mathrm{diag}( 1,1,1,-1)
\end{align}
on a pair of qubits. Since it is diagonal, it commutes with any other $U_{CZ}$ acting on other qubits.
This commutes with Pauli $Z$, so only nontrivial action is on Pauli $X$.
It is simple calculation to verify that $U_{CZ} (X \otimes I) U_{CZ}^\dagger = X \otimes Z$ and
$U_{CZ} (I \otimes x) U_{CZ} = Z \otimes X$.
This implies that if $U_{CZ}$ acts on every pair of qubits separated by the displacement $m \in \ZZ^D$,
then the induced symplectic transformation is precisely $S_{m + \bar m}$.

When $q \ge 2$, we can generalize the controlled-NOT \eqref{CNOT}.
If we apply the controlled-NOT translation invariantly 
where the target qubit is at $m \in \ZZ^D$
{\em relative} to the control qubit,
then the presence of $X$ at the control qubit will bring a new $X$ to the target qubit,
and the presence of $Z$ at the target qubit will bring a new $Z$ to the control qubit.
Since the controlled-NOT's on two different but overlapping pairs in general do not commute,
the control and target qubit should be the different qubit within the unit cell,
in order to define the controlled-NOT uniformly over the lattice unambiguously.
Such a translation-invariant controlled-NOT induces a symplectic transformation
\begin{align}
 C|_{\langle e_i , e_j , e_{i+n} , e_{j+n} \rangle } = C(m) :=
\begin{pmatrix}
 1 & 0 & & \\
 m & 1 & & \\
  &  & 1 & -m \\
  &  & 0 & 1
\end{pmatrix} \text{ where }~~1 \le i \neq j \le n
\label{CNOT2}
\end{align}
where $m$ is a monomial of $R$.
Since $C(m)C(m') = C(m+m')$,
the entry of $m$ in the matrix \eqref{CNOT2} 
can be occupied by any Laurent polynomial of $R$.

\begin{rem}
A natural question in analogy with \ref{SymplecticTransformationGenerators}
is whether the symplectic transformations that are found so far
generate the full symplectic group.
The author does not know the answer when $D \ge 2$.
The case $D=0$ is covered in \ref{SymplecticTransformationGenerators},
and the case $D=1$ will be solved in the next section.
\hfill $\diamond$ \end{rem}

\section{Low dimensions}

\subsection{Smith normal form}

For a moment, we digress from the translation invariant additive codes,
and consider matrices over a ``nice'' ring.
The conclusion will have immediate applications in one-dimensional additive codes.

A ring is an abelian group where two elements may be multiplied.
Integers, complex numbers, polynomials, etc. form rings.
An {\bf ideal} $I$ of a ring $R$ is a subset of the ring
that is a subgroup of $R$ under addition
such that $r m \in I $ for all $r \in R$ and $m \in I$.
So, an ideal is a collection of multiples of its members,
and sums thereof.
Most of the time, we only care about the generators,
the multiples of which form the ideal.
When the generators $g_i$ are known, we write $I = (g_i,\ldots)$.
A commutative ring with 1 is a {\bf principal ideal domain}
if no nonzero elements multiply to become zero 
and every ideal is generated by a single element.
Important examples are the ring of integers $\ZZ$,
and the (Laurent) polynomial ring in one variable over a field $\FF[x]$ ($\FF[x,x^{-1}]$).
These are examples of {\bf Euclidean domains}.

Why is $\ZZ$ a principal ideal domain?
Suppose $I = (a,b)$ is an ideal of $\ZZ$.
If $b > a > 0$, then we can consider $(a, b-a)$,
and see that it is the same ideal as before 
because an ideal is closed under addition.
We can proceed similarly to find smaller and smaller generators,
but this must stop at some point because positive integers cannot decrease forever.
In the end, we must be left with a single number (principal generator),
which is actually the greatest common divisor of $a$ and $b$.
This is the Euclid's algorithm, and proves that $(a,b) = (\gcd(a,b))$.
This argument will always give the single generator of 
an ideal of $\ZZ$ generated by a finitely many elements.
For an arbitrary ideal $I$ where the number of generators is unknown,
one can consider the smallest positive member $d$ of the ideal,
and prove that every member of $I$ has to be a multiple of $d$.
Important is the division algorithm 
that allows us to find a {\em smaller} element 
as a linear combination of two given elements.
This argument applies with a slight modification
to the univariate polynomial and Laurent polynomial ring.
The ``size'' of a polynomial is measured by the degree of the polynomial,
and the size of a Laurent polynomial is measured by
the difference of the greatest exponent and the least exponent.
We can rephrase the implication of the algorithm as follows.
\begin{prop}
For any column vector $v$ over an Euclidean domain,
there exists a finite product $M$ of elementary row operation matrices
such that the column vector $Mv$ consists of a single nonzero entry.
\end{prop}

\subsubsection{Classification of finitely generated abelian groups}

A group $G$ is said to be {\bf finitely generated}
if there is a finite set of members (generators) $g_1, \ldots, g_n$
of the group such that all other members
can be written as a finite product of the generators and their inverses.
Following the convention for abelian groups,
we should denote the group operation as a sum instead of a product.
So, any group element can be written as
\begin{align}
 \sum_{i=1}^n c_i g_i
 \label{eq:LinearCombination}
\end{align}
where $c_i \in \ZZ$. This means that the $n$-tuple of integers $(c_1,\ldots, c_n)$
can express any element of the group, and we may say that
the map
\begin{align}
 \ZZ^n \to G
\end{align}
is surjective.
This map is a group homomorphism as one can easily check from \eqref{eq:LinearCombination}.
Whenever we see a surjective homomorphism we should consider the kernel $K$.
The kernel itself is finitely generated%
\footnote{
This is not too trivial.
One of the best ways to show this is through the notion of {\bf Noetherian} rings and modules.
Here is a sketch.
A Noetherian module is one of which any submodule is finitely generated.
An equivalent definition is that every increasing chain of submodules
saturates.
Similarly, a Noetherian ring is Noetherian if it is Noetherian module over itself.
After showing the equivalence of the definitions,
one can further show that finitely generated modules over a Noetherian ring is Noetherian.
Since $\ZZ$ is a principal ideal domain, it is Noetherian.
The kernel is a submodule of $\ZZ^n$ and therefore is finitely generated.
}
by $k_1,\ldots, k_m \in \ZZ^n$.
We again think of the kernel as the image of the map
\begin{align}
 \varphi : \ZZ^m \to \ZZ^n
\end{align}
sending the unit vectors of $\ZZ^m$ to $k_i$.
In this way we \emph{present} the group $G$ as the cokernel of $\varphi$:
\begin{align}
\ZZ^m \xrightarrow{\varphi} \ZZ^n \to G \to 0.
\end{align}
If we express the map $\varphi$ as a matrix, it will contain the generators $k_1, \ldots, k_m$
in the columns of the $n \times m$ integer matrix $M$.

Now, what do the row and column operations on $M$ correspond to?
Instead of choosing the generating set of $G$ as $g_1,\ldots, g_n$,
we could choose $g_1+ g_2, g_2, g_3, \ldots, g_n$.
In the latter case, we would write \eqref{eq:LinearCombination} as
\begin{align}
 g = c_1(g_1+g_2) + c_2 g_2 + \cdots c_n g_n = c_1 g_1 + (c_1 + c_2) g_2 + \cdots c_n g_n 
\end{align}
We see that different choices of the generators lead to row operations on $M$.
Similarly, the choice of generators of $K$ is of course arbitrary,
and since the columns of $M$ are the generators of $K$,
this corresponds to column operations on $M$.

Therefore, any row and column operations on $M$ does not change 
the isomorphism class of $G$;
they are just differences how we describe the group $G$.
Let us use the Euclid's algorithm in order to simplify $M$.
Pick any nonzero column and run the Euclid's algorithm
to eliminate all but one entries in the upper-left corner.
Run the algorithm on the first row, to single out a nonzero entry on the upper-left corner.
If any nonzero element appears, repeat.
Since a positive integer cannot decrease forever,
this procedure must end after finitely many iterations.
A nonzero element will reside at the upper-left corner
and all other entries in the first column and first row will be zero.
\begin{align}
 \begin{pmatrix}
  \star & \star & \star \\ \star & \star & \star \\ \star & \star & \star 
 \end{pmatrix} \mapsto
 \begin{pmatrix}
  d_1 & 0 & 0 \\ 0 & \star & \star \\ 0 & \star & \star 
 \end{pmatrix}
\end{align}
If the bottom-right block contains an integer that is not divisible by $d_1$,
then we bring that integer to the first row, and repeat the above.
This will decrease the number in the upper-left corner,
and after finitely many iterations,
$d_1$ will divide all the other entries.

We inductively proceed to obtain
\begin{align}
 \begin{pmatrix}
  d_1 & 0 & 0 \\ 0 & \star & \star \\ 0 & \star & \star 
 \end{pmatrix} \mapsto
  \begin{pmatrix}
  d_1 & 0 & 0 \\ 0 & d_2 & 0 \\ 0 & 0 & \ddots 
 \end{pmatrix}
\end{align}
where
\begin{align}
 d_1 ~|~ d_2 ~|~ \cdots ~|~ d_{r}.
\end{align}
where $d_r > 0$ and $r$ is the rank of the matrix $M$,
i.e., $r$ is maximal such that some $r \times r$ submatrix of $M$ has nonzero determinant.
This diagonal form is called the {\bf Smith normal form} of $M$.
The diagonal elements are called {\bf elementary divisors} of $M$.
We have defined the elementary divisors as a result of the Smith algorithm.
There is some arbitrariness in the details of the algorithm.
So, a priori, we do not know whether the elementary divisors
are unique regardless how we obtain them.
However, there is another characterization of elementary divisors,
which will prove their uniqueness.
For any rectangular matrix $M$, let $I_s(M)$ be 
the ideal generated by the determinants of $s \times s$ submatrices of $M$,
called {\bf $s$-th determinantal ideal} of $M$.
\begin{prop}
It holds that $I_s(M) = I_s(AM) = I_s(MB)$ for any invertible matrices $A$ and $B$.
The elementary divisors of $M$ are determined by the determinantal ideals of $M$,
and hence are uniquely determined by $M$.
\end{prop}
\begin{proof}
The second claim follows from the first because
$d_1 \cdots d_s$ is the (principal) generator of $I_s(M)$.

To show the first claim, suppose three matrices satisfy $AM=C$.
An ($s \times s$)-submatrix $C'$ of $C$ is the product of some $s \times s'$ submatrix $A'$ of $A$ 
and some $s' \times s$ submatrix $M'$ of $M$.
Each row $v_i$ is a linear combination of rows $b_j$ of $M'$ as $v_i = \sum_{j_i} A'_{i j_i} b_{j_i}$.
The minor $\det C'$ is a multilinear function of rows $v_i$ of $C'$;
\begin{align}
\det C' 
&= \det( v_1, \ldots, v_s ) \\
&= \sum_{j_1, j_2,\ldots, j_s} A'_{1 j_1} A'_{2 j_2} \cdots A'_{s j_s} \det( b_{j_1}, \ldots, b_{j_s} ) .
\end{align}
This implies that $\det C'$ is a linear combination of minors of $M'$.
Since $C'$ was arbitrary, we see that $I_s(C) \subseteq I_s(M)$.
If $A$ is invertible, then the opposite inclusion holds, implying an equality.
This proves $I_s(AM) = I_s(M)$.
To show $I_s(MB) = I_s(B)$, transpose everything above.
\end{proof}
Since $G$ is the cokernel of $M$,
we see that
\begin{align}
 G \cong \ZZ/d_{t+1} \ZZ \oplus \cdots \oplus \ZZ / d_r\ZZ \oplus \ZZ^{n-r} .
\end{align}
where $d_t = 1 < d_{t+1}$.%
\footnote{
To show that this is a unique expression,
observe the following.
The number of the direct summands $\ZZ$ is the vector space dimension
upon tensoring $\mathbb Q$ over $\ZZ$.
In addition, for the minimal $n$, either $d_1 = 0$ or $d_1 > 1$.
}
We have decomposed the group $G$ into familiar abelian groups,
effectively by row and column operations on the presentation of $G$.

\begin{rem}
Over any principal ideal domain the Smith normal form
is defined: For any matrix $M$ there exists invertible matrices $A$ and $B$
such that $A M B$ is diagonal such that upper-left elements divide
lower-right elements.
The difference is that over a non-Euclidean domain,
the algorithm to find $A$ and $B$ 
may need to do more than just adding one row (column) to another row (column).
The conclusion that the Smith normal form is unique (i.e., the elementary divisors are well-defined)
remains true.
\hfill $\diamond$ \end{rem}

\subsection{Classification in one dimension}

We turn back to the lattice codes with translation structure,
here with the simplest possible translation.
Consider qubits arranged on a straight line
where $q$ qubits are clustered at each site $i \in \ZZ$.
The translation group is $\ZZ$ 
and we identify the group algebra as $R=\FF_2[x,x^{-1}] \cong \FF_2 [\ZZ]$.
We have shown that
a translation invariant additive code on this array of qubits is defined by
a stabilizer module $\Sigma$ over $R$.
The stabilizer module has finitely many generators as an $R$-module,
and if we express the generators in the columns of a $2q \times t$ matrix $\sigma$
then it satisfies
\begin{align}
 \sigma^\dagger \lambda_q \sigma = 0
\end{align}
by \ref{IsotropyCondition}.
The entries of $\sigma$ are Laurent polynomials with one variable $x$.
The ring $R$ happens to be a Euclidean domain,
and we can try to convert the matrix $\sigma$ into a simpler form
by the elementary symplectic transformations that we found in Section~\ref{sec:LatticeSympTrans}.
We will show that the Smith normal form of $\sigma$ can be obtained.
\begin{prop}
For any translation-invariant one-dimensional additive code,
there exists a locality-preserving and translation-invariant 
Clifford transformation such that
the stabilizer map is diagonal with zero matrix in the bottom half.
The diagonal elements completely determines the equivalence class
of the translation-invariant additive code up to Clifford operations.
\label{1Dpreclassification}
\end{prop}
This has appeared in \cite{GrasslRoetteler2006Convolutional} and also in \cite{Haah2012PauliModule}.
\begin{proof}
We employ the technique used in the proof of \ref{IsotropicSubspaceComplement}.
Using the controlled-NOT, Hadamard, and column operations,
we can bring $\sigma$ into
\begin{align}
\sigma' =
 \begin{pmatrix}
  f & 0 \\
  0 & \star \\
  g & \star \\
  0 & \star
 \end{pmatrix}.
\end{align}
The equation $\sigma'^\dagger \lambda_q \sigma' = 0$ demands that $f \bar g = \bar f g = \overline{f \bar g}$.
Suppose $f = \alpha x^a + \cdots + \beta x^b$ and $g = \gamma x^c + \cdots + \delta x^d$ with the exponents increasing
and $\alpha,\beta,\gamma,\delta$ are all nonzero.
Assume harmlessly that the degree of $f$ is smaller than that of $g$: $b-a < d-c$.
(If not, apply Hadamard to interchange them).
Then the equation $f \bar g = \overline{f \bar g}$ 
implies that $\gamma/\alpha = \delta / \beta$ and $a+b = c+d$.
Setting $h = (\gamma/\alpha)(x^{c-a} + x^{d-b})$, we see that $h = \bar h$
and controlled-Z by $h$ reduces the degree at the position of $g$ by at least 2.
Since the degree cannot decrease forever, 
after finitely many iterations
we obtain a column with sole nonzero entry.
Then, the top row of the bottom half block must be zero, 
due to the equation $\sigma \lambda_q \sigma = 0$.
\begin{align}
\sigma'' =
 \begin{pmatrix}
  d_1 & 0 \\
  0 & \star \\
  0 & 0 \\
  0 & \star
 \end{pmatrix}.
\end{align}
If there is any entry that is not divisible by $d_1$, then we can bring it to the first row by Hadamard and controlled-NOT,
and repeat the above.
The degree cannot decrease forever, and we must be left with a stabilizer map in the form of $\sigma''$
where $d_1$ divides every entry.
We finish by induction in $q$.
\end{proof}

\subsubsection{Coarse-graining}

If we are lenient about the translation structure,
then stronger classification can be obtained.
The translation group $\ZZ$ has subgroups $b\ZZ$ for any positive integer $b$.
Any stabilizer module is a module over this smaller translation group,
and we can consider Clifford operations that conforms
with this smaller translation group.
This can be viewed as taking a larger unit cell in the lattice.
Instead of saying that the unit cell consists of $q$ qubits,
we now take the unit cell to consist of $bq$ qubits.

To be clear, by {\bf coarse-graining}
we mean taking a smaller base ring 
$R' = \FF_2[x^b,x^{-b}]$ of $R = \FF_2[x, x^{-1}]$,
and regarding all $R$-modules to $R'$-modules.

As a standalone ring, $R'$ is isomorphic to $R$,
but now $R$ is a module of rank $b$ over $R'$ with a basis $1,x,\ldots, x^{b-1}$.
The Pauli module $R^{2q}$ is now $R'^{2bq}$ as $R = R'^{b}$.
Since the stabilizer map is from $R^t$ to $R^{2q}$,
under the coarse translation group, 
the new stabilizer map is from $R'^{bt}$ to $R'^{2bq}$,
and the corresponding matrix gets bigger by a factor of $b$.
To figure out the bigger matrix,
observe that each entry in the stabilizer map can be regarded as a map $R \to R$.
Over the smaller ring $R'$, this map has to be represented by $R'^b \to R'^b$.
Since the multiplication in $R$ is compatible with the composition of maps $R \to R$
(That is what module is about),
it is enough to find the matrix representation of the generator $x$ of the ring over $\FF_2$.
The multiplication by $x$ sends the basis elements $1,x,\ldots, x^{b-1}$ to $x, x^2, \ldots, x^b = x' \cdot 1$.
Hence,
\begin{align}
 (x : R'^b \to R'^b) = 
 \begin{pmatrix}
  0 &   & & x' \\
  1 & 0 & &    \\
    & \ddots & \ddots & \\
   & & 1 & 0
 \end{pmatrix} .
\end{align}

\begin{exc}
Consider the {\bf Ising model}
whose stabilizer map is given by 
$\sigma_\text{Ising} = \begin{pmatrix} 1+x \\ 0 \end{pmatrix}$
on the one-dimensional lattice.
Write down the corresponding Hamiltonian according to \eqref{eq:Hamiltonian}.
Take a smaller translation group $2\ZZ$ instead of $\ZZ$,
and rewrite a corresponding stabilizer map, which should be $4 \times 2$.
Verify that the Hamiltonian is not changed.
\label{IsingModel}
\hfill $\diamond$ \end{exc}

Using this passive coarse-grain procedure,
the Clifford group actually becomes larger.
From \ref{1Dpreclassification}, the only potentially interesting stabilizer map is $(f)$
where we omit the lower half block and $f$ is a Laurent polynomial.
Since multiplying a monomial to $f$ does not change $\im (f)$ at all,
we may assume that all exponents of $f$ is nonnegative and $f(x=0) \neq 0$.
Upon coarse graining, this $1 \times 1$ matrix becomes $b \times b$ matrix,
and a controlled-NOT provides any row operation,
and redefinition of stabilizer generators provides any column operation.
Hence, this $b \times b$ matrix can be brought into the Smith normal form,
and we would hopefully simplify $f$ into a smaller degree polynomial in $x'$.
The following tells us how to choose $b$.
It is convenient to introduce the {\bf annihilator} of a module $M$ over a ring $S$:
\begin{align}
\ann_S M = \{ r \in S : rm = 0 ~\forall m \in M\}
\end{align}
An annihilator is an ideal of $S$.
\begin{prop}
Let $f(x) \in \FF[x]$ be a polynomial with $f(0) \neq 0$ over a field $\FF$.
Suppose $f(x)$ divides $x^n -1$.
Then, the annihilator of the module $M = \FF[x]/(f(x))$ over $R' = \FF[x^n]$
is precisely the ideal $(x^n - 1) \subseteq R'$.
\label{SubringAnnihilator}
\end{prop}
\begin{proof}
Since $R'$ is a subring of $R = \FF[x]$,
we immediately have $\ann_{R'} M = R' \cap \ann_R M = R' \cap (f(x))$.
The latter includes $x^n -1$ by the supposition.
Hence, $(x^n -1) \subseteq \ann_{R'} M \subsetneq R'$,
but $(x^n -1)$ is maximal in $R'$ since $R'/(x^n -1) \cong \FF$ is a field.
\end{proof}
Let us see how this implies the simplification of $f$ 
into a smaller degree polynomial in $x' = x^n$.
By coarse-graining, we obtain 
the matrix representation $M$ of the map $f(x) : R'^n \to R'^n$.
The $R'$-module $R / (f(x))$ is equal to $R'$-module $\coker M$,
and hence $\ann_{R'} \coker M = \ann_{R'} R/(f(x)) = (x^n -1) = (x'-1)$.
By inspection of the Smith normal form,
every elementary divisors must divide the annihilator $x'-1$.
There are only two possible ways to divide $x'-1$:
either by $1$ or by $x'-1$. This implies that the elementary divisors
are either $1$ or $x'-1$.
Since the module $M$ has $\FF$-vector space dimension $\deg f$,
there are precisely $\deg f$ elementary divisors that are equal to $x'-1$
and $n-\deg f$ elementary divisors that are equal to $1$, up to scalars in $\FF$.

We work out an example explicitly.
Let $f(x) = 1 + x + x^2$, and then $x^3 -1 = (x-1)f(x)$,
so let $R' = \FF[x^3]$.
As a $R'$-linear map $R'^3 \to R'^3$, $f(x)$ becomes
\begin{align}
M =
\begin{pmatrix}
1 & x' & x' \\
1 & 1 & x' \\
1 & 1 & 1
\end{pmatrix}.
\end{align}
Applying row and column operations, we have
\begin{align}
M \mapsto 
\begin{pmatrix}
1 & x' & x' \\
0 & 1-x' & 0 \\
0 & 1-x' & 1-x'
\end{pmatrix} \mapsto
\begin{pmatrix}
1 & 0 & 0 \\
0 & 1-x' & 0 \\
0 & 1-x' & 1-x'
\end{pmatrix}\mapsto
\begin{pmatrix}
1 & 0 & 0 \\
0 & 1-x' & 0 \\
0 & 0 & 1-x'
\end{pmatrix}.
\end{align}
In fact, the supposition in \ref{SubringAnnihilator}
is always satisfied whenever the field $\FF$ is finite,
which is the case for the additive codes on lattices.
\begin{prop}
For any polynomial $f(x)$ with coefficients in a finite field $\FF$ such that $f(0) \neq 0$,
there exists a positive integer $n$ such that $f(x)$ divides $x^n -1$.
\label{FFPolynomialFactor}
\end{prop}
\begin{proof}
We have to use some facts about finite fields;
namely, any finite field of characteristic%
\footnote{
The minimal positive integer $p$ such that
$p x = 0$ for any element $x$ of the field.
It is necessarily a prime number.
}
$p$ consists of solutions of $x^{p^m} -x = 0$.%
\footnote{
The proof using the result on finitely generated abelian groups can be found in
\ref{FieldEquation} below.
}
This means that the roots of the polynomial $f(x)$
are roots of $x^{p^m} -x$ for some $m$.
Since 0 is not a root of $f(x)$,
we see that the roots of $f(x)$ are among those of $x^{n'} -1$ for some $n'$.
If there is any multiplicity of the roots of $f(x)$,
take the smallest $m'$ such that $p^{m'}$
is at least the largest multiplicity.
Then, $(x^{n'} -1)^{p^{m'}} = x^{n' p^{m'}} -1$ contains
all factors of $f(x)$ and hence is a multiple of $f(x)$.
\end{proof}

Combining \ref{1Dpreclassification}, \ref{SubringAnnihilator}, \ref{FFPolynomialFactor},
we arrive at the classification:
\begin{thm}
Any one-dimensional translation-invariant additive code
can be converted into several copies of Ising models
and some trivial codes,
by Clifford operators that obey coarse translation invariance.
\label{OneDClassification}
\end{thm}

Note the consistency with our result \ref{OneDCodeDistanceBound}.
The Ising model has code distance 1
independent of system size.
In \ref{OneDCodeDistanceBound}, we did not assume the translation-invariance,
and concluded that the code distance is bounded.
Here we assumed the translation-invariance and obtained a complete classification
under Clifford operations.
The Clifford operations in \ref{OneDClassification}
preserve locality,
so any logical operator in the one-dimensional translation-invariant code
is a conjugation of a logical operator of the Ising model,
which acts on a geometrically local set of qubits.

\begin{rem}
The set $\FF^\times$ of all nonzero elements of a finite field
is a (multiplicative) group of finite order.
In particular, it is finitely generated and finite.
Therefore, it is isomorphic to direct sum of finite cyclic groups.
Choose the largest period $N$, which is a factor of $|\FF^\times|$.
Then, every nonzero field element satisfies $x^N = 1$.
The polynomial $x^N-1$ has at most $N$ distinct roots,
but we know it has $|\FF^\times|$ distinct roots.
Hence, $|\FF^\times| = N$ and the $\FF^\times$ is a cyclic group of order $N$.
Finally, $\FF$ is a vector space over $\FF_p$ of some finite dimension $m$,
implying $N = p^m -1$.
Therefore, $x^{p^m} -x = 0$ holds for any element of the field $\FF$.
\label{FieldEquation}
\hfill $\diamond$ \end{rem}

\subsection{Translation-invariant two-dimensional CSS codes}

In one-dimensional classification,
it was crucially used in the initial stage 
that the base ring $\FF_2[x, x^{-1}]$ is a Euclidean domain.
In two-dimensions, the base ring is $R = \FF_2 [x^\pm, y^\pm]$,
and the problem becomes more complicated.

First, we need to distinguish infinite lattice $\ZZ^2$
versus finite lattice $\ZZ^2 / \Lambda$ obtained by
periodic boundary conditions
(factoring out a subgroup $\Lambda \le \ZZ^2$ of finite index).
This distinction was unnecessary in the one-dimensional case
because the ring $\FF_2[x^\pm]$ was so simple
that we didn't have to talk about logical operators.
Roughly speaking, we saw that
a one-dimensional translation-invariant code
is not going to be useful for error correcting purposes,
which makes the discussion of logical operators unimportant.
On the contrary, 
two-dimensional codes may have error correcting capability
with the code distance comparable with the system size,
and it is crucial to understand logical operators.
Unfortunately,
the logical operators in the infinite system are
out of our scope since we have only studied finite dimensional symplectic spaces.

The complication is
already lurking in the previous one-dimensional case.
Consider the Ising model (see \ref{IsingModel}) 
with excitation map $\epsilon = (0, ~1+x )$.
Applying \ref{LogicalOperatorExcitationMap},
we would say that the logical operators correspond to $\ker \epsilon$.
Since in $\FF_2[x^\pm]$ no two nonzero elements multiply to become zero,
the kernel has the zero second component.
However, if we had considered $\ker \epsilon$ over the factor ring
$\FF_2[x^\pm] / (x^L -1) = \FF_2[x]/(x^L-1)$,
then the kernel would have nonzero second component 
since $(1+x)\sum_{i=1}^L x^i = 0$.
Remark here that the ideal $(x^L-1)$ imposes the \emph{periodic boundary condition}
that translation by $L$ units is equivalent to no translation.
Thinking of ``infinite $L$'' to recover the infinite lattice,
we would say that an infinite series $\sum_{i \in \ZZ} x^i (0,1)^T$ lies in $\ker \epsilon$,
but the infinite series is \emph{not} a member of the Pauli module $(\FF_2[x^\pm])^2$.
It is not too difficult to extend the module 
to include the infinite series, but we are not going to do it.

The most important distinction for the two- or higher dimensional
cases from the one-dimensional ones
is that there exists a stabilizer map $\sigma$ such that
\begin{align}
\ker \epsilon &= \im \sigma 
   \quad \text{ over } R = \FF_2[x^\pm, y^\pm], \label{eq:Exactness}\\
\ker \epsilon &\supsetneq \im \sigma
  \quad \text{ over } \FF_2[x,y]/(x^L-1,y^L-1) \label{eq:NontrivialLogical}
\end{align}
for some $L$, where $\epsilon = \sigma^\dagger \lambda_q$
is the excitation map.
In \ref{LogicalOperatorExcitationMap}
one has shown that $\ker \sigma^\dagger / \im \sigma$ 
is the set of all independent logical operators.
\eqref{eq:NontrivialLogical} means that there are
nontrivial logical operators,
and they cannot be expressed by a Laurent polynomial vector
in the infinite system.
The logical operators have to be \emph{global}.

If \eqref{eq:Exactness} does not hold,
then there exists a nontrivial logical operator
supported on a finite region of the lattice.%
\footnote{
We had better be more cautious here.
We only have defined logical operators in the finite systems,
and here we are saying that $\ker \epsilon / \im \sigma \neq 0$
over $R$ implies that there is a nontrivial logical operator for finite systems.
We elaborate on this in \ref{NonvanishingModule} below.
}
The size of the finite region is independent of the system size,
and an error on that finite region will not be corrected.
This is a situation we want to avoid.
Thus, we assume \eqref{eq:Exactness} from now on.

Here is an example of \eqref{eq:NontrivialLogical},
which we refer to as {\bf toric code}~\cite{Kitaev2003Fault-tolerant}.
\begin{align}
\sigma =
\begin{pmatrix}
x-1 & 0 \\
y-1 & 0 \\
0 & \bar y-1 \\
0 & -\bar x+1
\end{pmatrix}.
\label{eq:ToricCode}
\end{align}
\begin{exc}
Verify that $\sigma^\dagger \lambda_2 \sigma = 0$,
and that $\ker \sigma^\dagger \lambda_2 = \im \sigma$.
\hfill $\diamond$ \end{exc}
One could directly compute $\ker \sigma^\dagger \lambda_2 / \im \sigma$
over $R/(x^L-1, y^L-1)$,
but we are going to make connection from this quotient module 
to the cellular homology in the next subsection.
In the rest of this section we show that 
this is essentially the only example of translation-invariant additive (CSS) codes
in two dimensions.

\subsubsection{Canonical form of stabilizer maps}

The following fact is our starting point of the further discussion.
Unfortunately the proof is beyond the scope of this lecture note.
The proof can be found in Ref.~\cite{Haah2012PauliModule}.
\begin{prop}
For any two-dimensional translation-invariant additive code,
if the stabilizer map $\sigma$ satisfies
$\ker \sigma^\dagger \lambda = \im \sigma$ over $R = \FF_2[x^\pm, y^\pm]$,
then there is a choice of another stabilizer map $\sigma'$ such that
\begin{align}
 \im \sigma = \im \sigma' = \ker \sigma'^\dagger \lambda_q,
 \quad \ker \sigma' = 0,\label{eq:InjectiveSigma}
\end{align}
Moreover, any such $\sigma'$ has size $2t \times t$ for some $t$,
and there exists a positive integer $b$ such that
\begin{align}
 \ann_{R'} \coker \sigma'^\dagger = (x^b -1, y^b -1) \label{eq:AnnihilatorCondition}
\end{align}
where $R' = \FF_2[x^{\pm b}, y^{\pm b}]$
is the coarser translation group algebra.
\label{Canonical2Dsigma}
\end{prop}

Due to this,
we may assume that our stabilizer map is given such that $b =1$ and $\sigma' = \sigma$.
The last condition \eqref{eq:AnnihilatorCondition} 
imposes stringent restrictions on $\sigma$.
To see this, let us recall the definition of the annihilator.
The cokernel is $R^t / \im \sigma^\dagger$.
If $e_1, \ldots, e_t$ are the unit basis vectors of $R^t$,
the annihilator condition says that $(x-1)e_i$ and $(y-1)e_i$ 
must be in the image of $\sigma^\dagger$ whenever $e_i$ is nonzero modulo $\im \sigma^\dagger$.
In other words, a linear combination of the columns of $\sigma^\dagger$
must yield $(x-1)e_i$, and another yield $(y-1)e_i$ if it cannot generate $e_i$.
In particular, the row $i$ of $\sigma^\dagger$ must generate either
the maximal ideal $(x-1,y-1)$ or the unit ideal (the ring $R$ itself).

The notion of torsion submodules is useful to characterize the $\coker \sigma^\dagger$.
A {\bf torsion submodule} $\mathcal T(M)$ of a module $M$ over $R$
is defined%
\footnote{This definition assumes $R$ has no two nonzero elements
that multiply to become zero. That is, we are using the fact that
$R$ is a {\bf domain}.
Note that over a field a torsion submodule is always zero.
}
as
\begin{align}
\mathcal T(M) = \{~ m \in M ~|~ \exists r \in R \setminus \{0\} \text{ such that } rm = 0 ~\} .
\end{align}
The condition \eqref{eq:AnnihilatorCondition} says
any nonzero element of $\coker \sigma'^\dagger$
is a torsion element; $\mathcal T \coker \sigma'^\dagger = \coker \sigma'^\dagger$
is a torsion module.
Depending on the choice of $\sigma$,
$\coker \sigma^\dagger$ may not be a torsion module.
However,
\begin{prop}
If $\im \sigma = \im \tau$, then 
$\mathcal T \coker \sigma^\dagger$ and $\mathcal T \coker \tau^\dagger$
are isomorphic as $R$-modules.
\label{IsomorphicTorsion}
\end{prop}
\begin{proof}
Regard $\tau$ and $\sigma$ as matrices.
We may combine two matrices as $\mu = (\sigma~\tau)$.
By assumption, $\im \mu = \im \sigma = \im \tau$.
Since every column of $\tau$ is in the span of $\sigma$,
we can find a column operation matrix $C$ such that
$\mu C = (\sigma ~ 0)$.
Similarly, there is a column operation matrix $C'$ such that
$\mu C' = (0~\tau)$.
Now, $\coker (\mu C)^\dagger = \coker (C^\dagger \mu^\dagger)$,
where the invertible $C^\dagger$ induces an isomorphism between
$\coker \mu^\dagger$ and $\coker (C^\dagger \mu^\dagger)$.
It follows that $\coker (\sigma~0)^\dagger \cong \coker (0~\tau)^\dagger$,
and $\mathcal T \coker (\sigma~0)^\dagger \cong \mathcal T \coker (0~\tau)^\dagger$.
On the other hand,
a torsion submodule is oblivious to a free summand:
For any module $M$, we see $\mathcal T(M \oplus R) = \mathcal T(M)$
since $R$ is a domain.
To finish the proof, observe that 
$\coker (\sigma~0)^\dagger = (\coker \sigma^\dagger) \oplus R^m$
where $m$ is the number of columns of $\tau$,
so $\mathcal T \coker (\sigma~0)^\dagger = \mathcal T \coker \sigma^\dagger$,
and likewise $\mathcal T \coker (0~\tau)^\dagger = \mathcal T \coker \tau^\dagger$.
We conclude that $\mathcal T \coker \tau^\dagger \cong \mathcal T \coker \sigma^\dagger$.
\end{proof}

\subsubsection{Structure theorem}

The strong constraint \eqref{eq:AnnihilatorCondition}
leads to a structure theorem, at least for CSS codes.
Recall that a {\bf CSS code}~\cite{CalderbankShor1996Good,Steane1996Multiple}
is a code where stabilizer generators
can be chosen to be either $X$- or $Z$-type.
Hence, a CSS code has a block diagonal stabilizer map $\sigma$
and excitation map $\epsilon$.
\begin{align}
\epsilon = \sigma^\dagger \lambda = \begin{pmatrix}
0 & \sigma_X^\dagger \\ -\sigma_Z^\dagger & 0
\end{pmatrix}, \quad
\sigma = \begin{pmatrix}
\sigma_X & 0 \\ 0 & \sigma_Z 
\end{pmatrix}.
\end{align}
\begin{thm}
For any two-dimensional translation-invariant CSS code,
if the stabilizer map $\sigma$ satisfies
$\ker \sigma^\dagger \lambda = \im \sigma$ over $R = \FF_2[x^\pm, y^\pm]$,
then the code becomes a tensor product of finitely many copies of the toric code
and a product state by (a finite number of layers of) Clifford operations.
The number of copies of the toric code in the CSS code is equal to 
$\frac 12 \dim_{\FF_2} \mathcal T \coker \sigma^\dagger$.
\label{CSS2dClassification}
\end{thm}

Note that by \ref{IsomorphicTorsion}, 
the number $\dim_{\FF_2} \mathcal T \coker \sigma^\dagger$
depends only on $\im \sigma$.
A similar result is in Ref.~\cite{Bombin2011Structure}.
See \ref{BombinResult} below for comparison.

\begin{proof}
\ref{Canonical2Dsigma} provides us with a stabilizer map that satisfies \eqref{eq:InjectiveSigma}.
Thus, we may assume that our stabilizer map $\sigma$ satisfies \eqref{eq:InjectiveSigma}
and \eqref{eq:AnnihilatorCondition} with $b=1$.

We wish to convert a row of the excitation map $\epsilon = \sigma^\dagger \lambda_q$
that generates the maximal ideal 
\begin{align}
\mm = (x -1, y-1)
\end{align}
into one that has only two components as in \eqref{eq:ToricCode}.
We know that there exists a vector $p$ of Laurent polynomials
such that $\epsilon p$ is a vector with sole nonzero entry $x-1$, $y-1$, or $1$.
We wish to turn $p$ into a unit vector.
The transformation has to be induced from Clifford operation,
and in particular should preserve the symplectic form.
We are not going to follow this line
as we do not understand the symplectic group over $R$
well enough.%
\footnote{
Even the following problem, which is presumably simpler,
is fairly complicated.
Let $v$ be a vector over $S = \FF[x_1,\ldots,x_n]$, a polynomial ring with the coefficients in a field.
Suppose the entries of $v$ generate the unit ideal.
Is there an invertible matrix $M$ with entries in $S$ where $v$ is a column of $M$?
The answer is affirmative, known as Quillen-Suslin theorem.
See XXI.3.5 of \cite{Lang}.
}
But, the complication can be reduced by going to a coarser lattice.

Since we can always multiply monomials on the rows of $\epsilon$,
we may assume every entry has positive exponent.
If $n$ is the maximum exponent of the terms in $\epsilon$,
take $R' = \FF_2[x^{\pm n}, y^{\pm n}]$ as our new base ring
(Coarse-graing I).
Then every entry of the new excitation map $\epsilon'$
has exponent at most 1. That is,
every entry is a $\FF_2$-linear combination of $1,x',y',x'y'$.
Let us say that such $\epsilon$ is \emph{of degree one}.
If a row $i$ of $\epsilon$ generates $\mm$,
then its entries are members of $\mm$,
and hence they are $\FF_2$-linear combinations of $x-1,y-1,xy-1$.
Using symplectic transformations that does not involve any variables $x,y$,
we can eliminate all but at most three entries in the row $i$.
The number of survived entries must be either three or two,
because they should generate two elements $x'-1$ and $y'-1$ over $R$.
If there are three entries survived,
they can be rearranged to be $x'-1$, $y'-1$, and $x'y'-1$.%
\footnote{If the base field is $\FF_p$ for some prime $p$,
then one should multiply some scalar of $\FF_p$ to get the coefficient 1.}
Now, we can \emph{use the CSS assumption}
to erase the entry $xy-1$ by some controlled-NOT operation.
(Note that a symplectic transformation on $\epsilon$ acts on its columns.)
This controlled-NOT operation may
destroy the fact that $\epsilon'$ has been of degree one,
in which case we take an even coarser base ring to reduce
the maximum exponent back to $1$ (Coarse-graining II).

Therefore, we may now assume without loss of generality 
that if $\mathcal T \coker \epsilon \neq 0$,
there exists a row of $\epsilon$ such that it consists of two entries of degree 1.
There are three possibilities up to $\FF_2$-symplectic transformations.
The two entries can be (i) $x-1,y-1$, (ii) $x-1, xy-1$, or (iii) $y-1,xy-1$.
In either of (ii) and (iii),
we can redefine the translation variable so that $xy \mapsto y$ or $xy \mapsto x$.
They corresponds to automorphisms of $R$.
To summarize,
\begin{itemize}
\item[($\bullet$)] If $\mathcal T \coker \epsilon \neq 0$, then 
by symplectic transformations and a redefinition of the base ring,
we can find an excitation map $\epsilon$ with a row $i$
that has only two nonzero entries $x-1$ and $y-1$.%
\footnote{
This does not involve any additional qubits;
however, a similar statement can be shown much more simply
if one allows to insert additional qubits in the trivial product state.}
\end{itemize}

Let us say that $x-1$ is at column $a$ and $y-1$ at $b$.
Assume that $a,b \le q$;
otherwise, interchange left-half and right-half blocks of $\epsilon$:
\begin{align}
\epsilon = \left(\begin{array}{ccc|c}
\star & \star  & \star &\\
0 & x-1  & y-1  &\\
\hline
&&& \star~\star~\star
\end{array}\right).
\end{align}
Now we examine what can be done to $\epsilon$
towards the form of \eqref{eq:ToricCode}.
It is going to be elementary and careful
examination of the conditions \eqref{eq:InjectiveSigma} and \eqref{eq:AnnihilatorCondition}.
We state intermediate goals and then give the proofs for them.
Let's roll our sleeves up!
\begin{itemize}
\item[(I)] After some $\FF_2$-symplectic transformation,
 we can obtain $\epsilon'$ such that
 all the entries on columns $a$ and $b$ of $\epsilon'$ belong to $\mm$:
 \begin{align}
\epsilon' = \left(\begin{array}{ccc|c}
\star & \circ  & \circ &\\
0 & x-1  & y-1 &\\
\hline
&&& \star~\star~\star
\end{array}\right) \text{ where any } \circ \in \mm.
\end{align}
\end{itemize}
By the condition \eqref{eq:AnnihilatorCondition}
there exists a vector $p_x$ and $p_y$
such that $\epsilon p_x = (x-1)e_i$ and $\epsilon p_y = (y-1) e_i$,
where $e_i$ is a unit vector with $i$-th component 1.
Setting $x=y=1$, $p_x, p_y$ become a vector over $\FF_2$,
and their symplectic product is zero due to the CSS condition.
Then, by \ref{IsotropicSubspaceComplement} 
there exists a symplectic transformation $S$ (controlled-NOT's) 
that do not involve variables $x,y$
such that $S p_x(x=y=1) = e_a$ and $S p_y(x=y=1) = e_b$ 
while the row $i $ of $\epsilon S^{-1}$ is still the same as the row $i$ of $\epsilon$.
That is, $S p_x$ has $a$-th component outside $\mm$, 
but everything else in $\mm$.
Similarly, $S p_y$ has $b$-th component outside $\mm$, 
but everything else in $\mm$.
Now, the new excitation map $\epsilon' = \epsilon S^{-1}$ has the claimed property:
The equation $\epsilon' (Sp_x) = (x-1)e_i$ demands that
\begin{align}
 \epsilon'_{k,a} (Sp_x)_{a} + (\text{terms in }\mm) = 0 \text{ for } k \neq i,
\end{align}
implying $\epsilon'_{k,a} \in \mm$ for all $k \neq i$.
Similarly, $\epsilon' (S p_y) = (y-1) e_i$ demands that $\epsilon'_{k,b} (S p_y) + (\text{terms in }\mm) = 0$
whenever $k \neq i$, implying $\epsilon'_{k,b} \in \mm$ for all $k \neq i$.
\begin{itemize}
 \item[(II)] Both the columns $a+q$ and $b+q$ of $\epsilon'$ are nonzero.
\end{itemize}
If either of them was zero,
then the condition $\ker \epsilon = \im \lambda_q \epsilon^\dagger$
would imply that there must be a row with sole nonzero element 1 
at column $a$ or $b$.
That is not possible as we have shown in (I) that $\epsilon'$ 
has entries in $\mm$ for the columns $a$ and $b$.
\begin{itemize}
\item[(III)] For any row $j$ 
each tuple of entries $(\epsilon'_{j,a+q},\epsilon'_{j,b+q})$ is a $\FF_2$-multiple of the tuple $(x-xy, xy-y)$.
\end{itemize}
Since the row $i$ of ($\bullet$) has only two nonzero components,
the condition $\epsilon' \lambda_q \epsilon'^\dagger = 0$
implies that $\epsilon'_{j,a+q}(\bar x - 1) 
+ \epsilon'_{j,b+q} ( \bar y -1) = 0$ for all $j$.
The solution of degree-one is a $\FF_2$-multiple of $x-xy, ~xy-y$.
\begin{itemize}
\item[(IV)] By Gauss elimination that does not involve any variables $x$ and $y$,
the excitation map $\epsilon'$ can be turned into $\epsilon''$ such that each of
the left and right blocks of 
$\epsilon''(x=y=1)$
is in the reduced row echelon form.
\end{itemize}
This is obvious.
\begin{itemize}
\item[(V)] $\epsilon''$ has a row $j$ such that $\epsilon''_{j,a+q}, \epsilon''_{j,b+q} \neq 0$ but $\epsilon''_{j,k}(x=y=1) = 0$ for all $k$:
\begin{align}
\epsilon'' = \left(\begin{array}{ccc|ccc}
\star & \circ &  \circ &&&\\
0 & x-1  & y-1 & &&\\
\hline
&&& \star&\circ& \circ\\
&&&  \circ& x-xy & xy-y \\
\end{array}\right).
\end{align}
\end{itemize}
Let $J_1$ be the collection of all row indices $j$
such that the row $j$ is nonzero upon setting $x=y=1$.
The reduced row echelon form of $\epsilon''$ tells us that
$j \in J_1$ if and only if $e_j \in \im \epsilon''$,
which can be verified by setting $x=y=1$.
For $j \in J_1$, choose vectors $f^{(j)}$ such that $e_j = \epsilon'' f^{(j)}$.
($f^{(j)}$ are determined up to $\im \lambda_q \epsilon''^\dagger$.)
On the contrary to the claim (V), suppose that 
$\epsilon''_{j,a+q}, \epsilon''_{j,b+q} \neq 0$ happens only for $j \in J_1$.
Then, for the unit vector $e_{a+q}$ we have 
$\epsilon'' e_{a+q} = \sum_{j \in J_1} c_j e_j$
for some $c_j \in R$.
This means that $e_{a+q} - \sum_{j \in J_1} c_j f^{(j)} \in \ker \epsilon''$.
On the other hand, setting $x=y=1$, we see that $\epsilon''(x=y=1)e_{a+q} = 0$ by (III).
This implies that $c_j(x=y=1) = 0$ for all $j \in J_1$.
Since $\ker \epsilon'' = \im \lambda_q \epsilon''^\dagger$,
it follows that $\lambda_q e_{a+q} - \sum_{j\in J_1} c_j \lambda_q f^{(j)} \in \im \epsilon''^\dagger$,
and therefore $\lambda_q e_{a+q} = e_a \in \im \epsilon''^\dagger(x=y=1)$.
But, the column $a$ of $\epsilon''(x=y=1)$ is zero by (I).
This is a contradiction. Since (II) implies that the column $a+q$ and $b+q$ are nonzero,
(V) is proved.
\begin{itemize}
\item [(VI)] By symplectic transformations and $\FF_2$-row operations on $\epsilon''$,
we can obtain $\epsilon'''$ such that $\epsilon'''_{j,a+q}$ and $\epsilon'''_{j,b+q}$
are the only nonzero entries in the row $j$ ($\neq i$) and columns $a+q$, $b+q$:
\begin{align}
\epsilon''' = \left(\begin{array}{ccc|ccc}
\star & \star &  \star &&&\\
0 & x-1  & y-1 & &&\\
\hline
&&& \star&0& 0\\
&&&  0& x-xy & xy-y \\
\end{array}\right).
\end{align}
\end{itemize}
By (III), row operations on $\epsilon''$ can eliminate all the entries but $j$-th row 
in the columns $a+q$ and $b+q$.
The row $i$ is left intact as it has zero entries in the columns $a+q$ and $b+q$ anyway.
After this, since $x-xy, xy-y$ generate the ideal $\mm$ and all the entries of row $j$
are in $\mm$ by (V), appropriate controlled-NOTs clean the row $j$
by adding some multiple of $\epsilon''_{j,a+q}$ and $\epsilon''_{j,b+q}$ to other places in the row $j$
in the right-half block of $\epsilon''$.
These controlled-NOTs do not change the row $i$ because on the left-half block of $\epsilon''$
the controlled-NOT adds multiples of 0 to $\epsilon''_{i,a}$ and $\epsilon''_{i,b}$.
\begin{itemize}
 \item [(VII)] For any row $r$ the tuple $(\epsilon'''_{r,a},\epsilon'''_{r,b})$ is a $R$-multiple
 of the tuple $(x-1,y-1)$. Therefore, some row operation on $\epsilon'''$ clears up the column $a,b$,
 leaving $\epsilon'''_{i,a} = x-1$ and $\epsilon'''_{i,b} = y-1$ intact:
 \begin{align}
\epsilon'''' = \left(\begin{array}{ccc|ccc}
\star & 0 &  0 &&&\\
0 & x-1  & y-1 & &&\\
\hline
&&& \star&0& 0\\
&&&  0& x-xy & xy-y \\
\end{array}\right).
\end{align}
\end{itemize}
This follows from the proof of (III) using the cleaned row $j$ of $\epsilon'''$.

Now the excitation map $\epsilon''''$ obtained from (VII) has a distinguished submatrix
consisting of rows $i$ and $j$, columns $a,b,a+q,b+q$, which is the same the excitation map of the toric code.
We have shown that $\mathcal T \coker \epsilon \neq 0$ implies $k = \dim_{\FF_2} \mathcal T \coker \epsilon \ge 2$.
By induction on $k$,
we see that $k$ has to be even,
and that there is a Clifford operation obeying a coarse translation-invariance
which factors out $k/2$ copies of the toric code.
The proof so far is algorithmic.

It remains to turn the code state into a product state when $k = 0$.
This part appears to be no easier than the Quillen-Suslin theorem
on the unimodular completion problem.
Since the excitation map of a CSS code is block diagonal,
and any elementary column operation on the left-block can be compensated
by another column operation in the right-block to become a symplectic transformation,
we only need to show that the left-half block of $\epsilon$ can be transformed into 
the identity matrix, representing the trivial state.
The condition \eqref{eq:AnnihilatorCondition} says that
each row generates the unit ideal (each row is unimodular),
and therefore our problem is
to convert each row that generates the unit ideal into the unit vector
using elementary column operations (unimodular completion)
over the Laurent polynomial ring.
This is known to be possible with explicit algorithms.
Park~\cite{Park2004Symbolic} gives an algorithm that reduces the problem into that over a polynomial ring,
so that the algorithm of Park and Woodburn~\cite{ParkWoodburn1995Algorithmic} can then be applied.
(Amidou and Yengui~\cite{AmidouYengui2008algorithm} find a more efficient algorithm
when the coefficient field is infinite, which is not applicable for our purpose.)
This completes our proof of the classification theorem.
\end{proof}

\subsubsection{Remarks}

\begin{rem}\label{BombinResult}

Bomb\'in~\cite{Bombin2011Structure} reports a classification theorem,
which is similar to \ref{CSS2dClassification}.
Here we sketch his approach and clarify the difference from \ref{CSS2dClassification}.
In Ref.~\cite{Bombin2011Structure},
the content of \ref{Canonical2Dsigma} is stated,
albeit in different terminology.
Namely, it is argued that there exists a translation-invariant set 
of local generators of the stabilizer group such that there is no ``local relation''
among the generators.
This amounts to $\ker \sigma = 0$ in our language.
In addition, 
there are finitely many topological charges which form an abelian group,
and one can coarse-grain the lattice
such that each new unit cell can support any topological charge.
The former amounts to $\coker \epsilon$ being a finite additive group,
and the latter amounts to \eqref{eq:AnnihilatorCondition}.

Given this intermediate conclusion,
Bomb\'in studies the topological spins and mutual statistics
using the commutation relations among string operators
associated with the topological charges,
using an idea by Levin and Wen~\cite{LevinWen2003Fermions}.
His finding can be summarized as follows.
The mutual statistics can be captured by
a bilinear symmetric form $\kappa$
with the diagonal elements being zero.
Over the binary field $\FF_2$,
the symmetric form $\kappa$ happens to be symplectic,
and therefore there exists a canonical basis as in \ref{SymplecticSpaceStructure}
for the set of topological charges.
Under this canonical basis,
there are two complementary subsets $C, D$ of topological charges
on which $\kappa$ vanishes.
The topological spin $e^{i \pi \theta(a)}$
obeys 
$\theta( a - b ) = \kappa(a,b) + \theta(a) + \theta(b)$
for any charges $a$ and $b$~\cite[Eq.~(3.1.2)]{BakalovKirillov2001},~\cite[Eq.~(223)]{Kitaev2006Anyons}.
On a subspace $C$ or $D$ where $\kappa$ vanishes,
the angle $\theta$
becomes a linear functional $\theta_C$ or $\theta_D$.
(We used the property that $b = -b$.)
Then, there is a basis change on the space of topological charges
to simplify the linear functionals $\theta_C$ and $\theta_D$.
If these linear functionals cannot be transformed to zero on $C$ and $D$,
then Bomb\'in calls the code ``chiral.''

The main conclusion is that any non-chiral translation-invariant code is equivalent
to a finitely many copies of the toric code.
Here, the equivalence is defined through locality-preserving unitary transformations
that map any Pauli operator to a Pauli operator.
There is a difference between this equivalence and the equivalence
under controlled-NOT, Hadamard, and phase gates.
The latter equivalence (under local quantum circuits) trivially implies the former notion,
while the former does not imply the latter unless one includes auxiliary qubits.

Ref.~\cite{ArrighiNesmeWerner2007} shows that 
for any locality-preserving unitary $U$
there exists another locality-preserving unitary $V$
such that
$U \otimes V$ is a local quantum circuit.
The proof is very simple:
Let $\mathcal S$ denote the swap operation
between two identical many-qubit systems.
This is a local quantum circuit as the swap can be implemented by swapping every pair of qubits,
i.e., $\mathcal S = \bigotimes_i \mathcal S_i$.
Then, 
$(U^\dagger \otimes I) \mathcal S_i (U \otimes I)$ is a local unitary,
and therefore $(U^\dagger \otimes I) \mathcal S (U \otimes I)$ is a local quantum circuit,
and
$U \otimes U^\dagger = \mathcal S (U^\dagger \otimes I) \mathcal S (U \otimes I)$
is a local quantum circuit, too.
However, this result together with Bomb\'in's
does not imply that one can turn a 2D translation-invariant code state
with no topological charge
into a product state by a local quantum circuit,
since the auxiliary qubits may be transformed into an entangled state.
Our treatment under CSS assumption
does not need the auxiliary qubits.

Note that \ref{CSS2dClassification} 
remains true for qudit stabilizer codes with prime dimensions
where $X$-type stabilizer group generators are separated from $Z$-type ones (CSS).
On the other hand, Bomb\'in's result does not assume the CSS condition,
but appears to rely on the binary field $\FF_2$
over which a skew-symmetric matrix is symmetric.
This make it obscure to apply his approach to qudit prime-dimensional codes.
\hfill $\diamond$ \end{rem}

\begin{rem}\label{NonvanishingModule}
Here we elaborate on the relation between $\ker \epsilon / \im \sigma$ being nonzero
and logical operators on finite systems.
Suppose $v \in \ker \epsilon$ is a vector such that $v \notin \im \sigma$,
so when passed to $\ker \epsilon / \im \sigma$ it represents a nonzero element.
The equation $\epsilon v = 0$ is still valid 
even if we mod out $\mathfrak b_L = (x^L-1,y^L-1)$,
meaning that $v$ is a logical operator.
Here, we need to prove that $v$ is nontrivial for some $L$.
That is, $v \notin \im \sigma + \mathfrak b_L R^{2q}$ for some $L$
if $v \notin \im \sigma$.
Phrased differently, we have to show that 
$(R/\mathfrak b_L) \otimes_R K \neq 0$
where $K := (Rv + \im \sigma) / \im \sigma \neq 0$.

Here is a proof for readers who are familiar with localization.
(See Atiyah-MacDonald~\cite{AtiyahMacDonald} for concise treatment.)
Let us lift the coefficient field to the algebraic closure $\FF^a$ of $\FF_2$.
($\FF^a$ is a flat module over $\FF_2$.)
$K \neq 0$ if and only if $K_{\mm} \neq 0$ for some maximal ideal $\mm$.
By Nullstellensatz, $\mm = (x - a, y - b)$ for some $a,b \in \FF^a$,
and hence there exists a positive odd integer $L$ such that $a^L = 1 = b^L$.
(There are infinitely many such $L$.)
Since $K$ is finitely generated $R$-module, it admits a finite presentation
i.e., $K \cong \coker \phi$ for some matrix $\phi$ over $R$.
Upon localization at $\mm$,
we can think of a minimal $\phi$ where every entry is a member of $\mm$.
Then, this minimal $\phi$ becomes a zero matrix
upon factoring $(\mathfrak b_L)_{\mm}$
because $(\mathfrak b_L)_{\mm} = \mm _{\mm}$.
Since the minimal $\phi$ had nonzero number of rows,
it follows that $(R/\mathfrak b_L \otimes K)_{\mm}$ is nonzero.
Therefore, $(R/\mathfrak b_L) \otimes K$ is nonzero
for infinitely many $L$.
\hfill $\diamond$ \end{rem}

\subsection{Connection to cellular homology}

We have learned that any translation-invariant CSS code
($X$ stabilizers and $Z$ stabilizers are separated)
in two-dimensions with large code distance has to be a finitely many copies of
the toric code.
The stabilizer map of the toric code consists of two blocks
which are related by exchanging two rows and $x \leftrightarrow \bar x,~
y \leftrightarrow \bar y$.
The excitation map inherits the duality as well,
and thus we can focus on one block.
The equation $\ker \epsilon = \im \sigma$ decomposes
into two equivalent equations, one of which reads
\begin{align}
R^1 \xleftarrow{\partial_1 = \begin{pmatrix} x-1 & y-1\end{pmatrix}}
R^2 \xleftarrow{\partial_2 = \begin{pmatrix} y-1\\ -x+1\end{pmatrix}}
R^1.
\end{align}
It is readily verified that $\partial_1 \partial_2 = 0$ as matrices.
We promised earlier that we will calculate $\ker \epsilon / \im \sigma$
subject to the periodic boundary condition that is 
expressed by an ideal $\mathfrak b_L = (x^L-1,y^L-1)$.
Here we make a connection to the cellular homology of the two-dimensional torus.

Consider a torus obtained by gluing small squares.
One can imagine a big square with vertical sides identified with each other,
and the horizontal sides identified with each other,
and then slice the big square into smaller ones.
Suppose the number of small squares are $L$ in each direction
so there are $L^2$ squares in total.

Introduce an abelian group $C_2 = \FF_2^n$ 
where $n = L^2$ is equal to the number of squares.
We may visualize an element of $C_2$ as a configuration of bits written 
on the centers of the small squares.
Similarly, introduce another abelian group $C_1 = \FF_2^{2n}$
where $2n$ is the number of edges.
Each vector in $C_1$ is identified with an array of bits
written on the edges.
Now, define a group homomorphism $\nabla_2 : C_2 \to C_1$
by the rule that the bit 1 on a square is mapped to
four bits written on the four edges that surrounds the square.
If a collection of squares is given as an input to $\nabla_2$,
then the output from $\nabla_2$
is the boundary of the union of the squares.
Thus, it is legitimate to call $\nabla_2$ a boundary map.

Let us introduce a coordinate system for the squares and the edges.
Pick a square and declare it as the origin.
The coordinates for the small squares are defined modulo $L$.
To give the coordinates for the edges, 
we associate the vertical edge to the square on the right of the edge,
and the horizontal edge to the square above the edge.
For example, 
the boundary edges of the square at $(0,0)$ are $(0,0)_v$, $(0,1)_h$,
$(1,0)_v$, and $(0,0)_h$,
where the subscripts distinguishes the horizontal or vertical edges.
Following the trick to write the collection of tuples as a polynomial,
we can write the collection of edges as $1_v + x_v$ and $1_h + y_h$.
Even more compactly, we have $\begin{pmatrix} 1+y \\ 1+x \end{pmatrix}$.
We see that this matrix is equal to $\partial_2$ over $\FF_2$.

\begin{exc}
Verify that $\partial_1$ represents the boundary map from the group of edges to the group of vertices.
To which square do we have to associate a vertex to have a consistent coordinate system?
\hfill $\diamond$ \end{exc}

The kernel of the boundary map modulo the group of boundaries of one-higher dimensional cells
is called a {\bf homology} group.
There is one homology group at each ``dimension'':
Formally we set $\partial_3 = 0$ and $\partial_0 = 0$,
and define $H_i = \ker \partial_i / \im \partial_{i+1}$ for $i=0,1,2$.

We could define these homology groups without ever involving the Laurent polynomial ring,
which was rather a convenient extra feature allowed by the Cartesian array of small squares.
Although it is a another problem how one can identify the homology group,
there is no problem of defining the homology group given any tessellation of the space
into small polygons.
To little (or big) surprise,
it is very important that the resulting homology group is independent of the tessellation 
(cell decomposition).
More general homology theory~\cite{Hatcher2002book}
is beyond the scope of this note,
but one can think of a toy version where one square from our tessellation is divided into two triangles.
In this case one can show that the homology groups remain the same.
On top of this, one can show that any two tessellations can be refined by a process of such subdivisions
to become the same tessellation, which leads the independence of the homology group on tessellations.
The simplest tessellation would be to have only a single square to cover the entire space (the two-dimensional torus),
with vertices and edges properly identified.
This is to set $L=1$ in our earlier discussion,
in which case the homology groups are clearly seen as $H_0 = \FF_2$, $H_1 = \FF_2^2$, and $H_2 = \FF_2$.

\begin{exc}
What is the code space dimension of the toric code under periodic boundary conditions?
What is the code distance?
\hfill $\diamond$ \end{exc}

\subsection{Dimension three and beyond}

As of this writing, there is no known
structure theorem for stabilizer modules in three or higher dimensions.
In terms of pure commutative algebra this problem appears 
not too well guided, since under a fairly mild condition
any chain complex obtained by continued calculation of kernels
(free resolution) gives rise to stabilizer module with large code distances.
General one-dimensional and CSS two-dimensional cases 
were fortunately simple enough that we could handle directly.

From a physical viewpoint,
one can focus on the behavior of excitations.
We briefly defined the excitations as a flipped term of a Hamiltonian,
but we did not give any further attributes to the excitations.
The most important notion to begin with is that of topological excitations ---
an excitation that cannot be created alone from a ground state by any local operator.
The topological excitation is insensitive to local basis change (Clifford transformations),
relabeling of terms of the Hamiltonian, or the translation structure.
The set of all point-like topological excitations 
is identified with 
the torsion submodule of the cokernel of the excitation map~\cite{Haah2012PauliModule},
and it is no coincidence that we introduced the torsion submodule
in the structure theorem for 2D CSS codes.
It is a characteristic of a code that is invariant under various basis changes.

There is conceptual progress from our consideration of additive codes
and corresponding Hamiltonians,
which is surprising from a conventional perspective on topological order.
It has been for long taken granted that a point-like topological excitations
has well-defined ``statistics,''
a consistent number that the quantum mechanical state vector acquires
upon exchanging two such excitations.
Following our formalism,
it is easy to show that this is not always the case.
In three or higher dimensions, there are many systems
of which point-like topological excitations cannot have well-defined ``statistics''
since their motion
is ill-defined~\cite{Chamon2005Quantum,BravyiLeemhuisTerhal2011Topological,
Haah2011Local,VijayHaahFu2015New}.

To illustrate the point more clearly,
consider the toric code excitation map: $\epsilon = ( x-1,~ y-1)$,
the cokernel of which is $\FF_2[x^\pm, y^\pm]/(x-1,y-1) \cong \FF_2$.
An element of the cokernel is an equivalence class of flipped terms
in the Hamiltonian modulo those that can be flipped by local Pauli operators.
Thus, the element of $\coker \epsilon$ is precisely what we call a topological excitation.
The action from the translation group on $\coker \epsilon \cong \FF_2$
is trivial. This means that a topological excitation at one location
is connected to that at other location by some local operator,
implying that the ``motion'' is allowed by some other interaction.
The algebraic origin to this conclusion is that the excitation map have binomial
elements $x-1$ and $y-1$ in the image $\im \epsilon$.
In higher dimensions, there is no reason for the excitation map to have
a binomial term in the image,
in which case the notion of motion for topological excitations becomes ambiguous.
An example in three dimensions is given by 
$\epsilon_\text{cubic} =( 1+x+y+z,~1+xy+yz+zx )$, known as the cubic code~\cite{Haah2011Local}.
The reader is encouraged to show that $\epsilon_\text{cubic}$
cannot generate a binomial term
in the image. 
In fact, $\coker \epsilon$ does not have any binomial zero-divisor.

The behavior of topological excitations can be even richer.
In four or higher dimensions, there is a system
in which a single Hamiltonian term cannot be flipped
alone~\cite{DennisKitaevLandahlEtAl2002Topological},
but they have to appear as a line or in some other more complicated shape.
There is a way of defining line-like topological excitations if the Hamiltonian
consists of commuting terms~\cite{Haah2014invariant},
but a concise algebraic characterization for additive code systems is immature.

\begin{acknowledgments}
The author thanks 
Cesar Galindo-Martinez and Julia Plavnik
for their hospitality during the workshop in Bogot\'a, Colombia,
and H\'ector Bomb\'in for guiding along his paper~\cite{Bombin2011Structure}.
The author is supported by Pappalardo Fellowship in Physics while at MIT.
\end{acknowledgments}

%

\end{document}